\documentclass[lettersize,journal]{IEEEtran}
    \usepackage{multicol}

\usepackage{siunitx}

\usepackage{comment}
\usepackage{graphicx,color}
\usepackage{amsmath, amsthm, amsfonts, amssymb, amsbsy,nccmath}
\usepackage{mathtools}
\theoremstyle{boldthm}
\usepackage{tikz}

\newcommand{\myarrow}[1][-45]{%
  \mathrel{%
    \text{$
     \begin{tikzpicture}[baseline = -0.5ex]
        \draw[->] (-1,1) node -- (1,-1) node; 
    \end{tikzpicture}
    $}%
  }%
}%
\usepackage{algorithm}
\usepackage{enumerate}
\usepackage{lipsum}
\newtheorem{lemma}{Lemma}
\usepackage{textgreek}
\usepackage{textcomp}

\usepackage{algorithmic}
\DeclareMathOperator{\KL}{KL}
\DeclareMathOperator{\E}{\mathbb{E}}
\usepackage{bbm}
\usepackage[sort,compress]{cite}
\usepackage{epsfig}
\usepackage{epstopdf}
\usepackage{dsfont}
\usepackage{epstopdf}

\usepackage{sidecap, caption}
\usepackage{caption}
\usepackage{subcaption}
\theoremstyle{definition}
\newtheorem{definition}{Definition}

\theoremstyle{Property}

\theoremstyle{assumption}

\theoremstyle{proposition}
\newtheorem{proposition}{Proposition}[section]
\theoremstyle{corollary}

\usepackage[inline]{enumitem}   
\makeatletter
\newcommand{\inlineitem}[1][]{%
\ifnum\enit@type=\tw@
    {\descriptionlabel{#1}}
  \hspace{\labelsep}%
\else
  \ifnum\enit@type=\z@
       \refstepcounter{\@listctr}\fi
    \quad\@itemlabel\hspace{\labelsep}%
\fi}
\makeatother
\DeclareMathOperator{\Cov}{Cov}

\newcommand\norm[1]{\left\lVert#1\right\rVert}

\newtheorem{theorem}{Theorem}
\usepackage{url}

\newcommand{\beq}{\begin{equation}}
\newcommand{\eeq}{\end{equation}}

\newcommand{\m}{{\!\, -\!\,}}

\newcommand{\mZ}{\mathcal{Z}}

\newcommand{\mY}{{\mathcal Y}}

\DeclareMathOperator*{\argmax}{arg\,max}
\DeclareMathOperator*{\argmin}{arg\,min}

\def\adots{\mathinner{\mskip0mu\raise0pt\vbox{\kern7pt\hbox{.}}\mskip3mu
          \raise4pt\hbox{.}\mskip3mu\raise8pt\hbox{.}\mskip0mu}}

\newcommand{\tr}{\mbox{tr}}

\usepackage{bm}

\newcommand{\bmx}{{\bm x}}
\newcommand{\bmy}{{\bm y}}

\newcommand{\mX}{\mathcal{X}}

\newcommand{\mD}{\mathcal{D}}
\newcommand{\mbS}{\mathbb{S}}

\newcommand{\mbE}{\mathbb{E}}
\newcommand{\mP}{\mathcal{P}}

\newcommand{\mR}{\mathcal{R}}

\newcommand{\mN}{\mathcal{N}}

\newcommand{\mM}{\mathcal{M}}

\newcommand{\bmzh}{\widehat{\bmz}}
\newcommand{\bmyh}{\widehat{\bmy}}

\newcommand{\btheta}{\boldsymbol{\theta}}

\newcommand{\mbI}{\mathbb{I}}

\makeatletter
\newcommand\fs@spaceruled{\def\@fs@cfont{\bfseries}\let\@fs@capt\floatc@ruled
  \def\@fs@pre{\vspace{0.5\baselineskip}\hrule height.8pt depth0pt \kern2pt}%
  \def\@fs@post{\kern1pt\hrule\relax}%
  \def\@fs@mid{\kern2pt\hrule\kern2pt}%
  \let\@fs@iftopcapt\iftrue}
\makeatother

\newcommand{\bit}{\begin{itemize}}
\newcommand{\eit}{\end{itemize}}

\newcommand{\mK}{\mathcal{K}}
\newcommand{\mL}{\mathcal{L}}

\newcommand{\mC}{\mathcal{C}}

\newcommand{\bmu}{\boldsymbol{u}}
\newcommand{\bmz}{\boldsymbol{z}}

\newcommand{\bmX}{{\boldsymbol{X}}}

\newcommand{\bms}{{\boldsymbol s}}

\DeclarePairedDelimiter\abs{\lvert}{\rvert}%

\usepackage{lipsum}

\newcommand{\xdownarrow}[1]{%
  {\left\downarrow\vbox to #1{}\right.\kern-\nulldelimiterspace}
}

\makeatletter

\newcommand\longleftrightarrowfill@{%
  \arrowfill@\leftarrow\relbar\rightarrow}
\makeatother

\restylefloat{algorithm}

\usepackage{flushend}
\begin{document}
\title{ Compositional Semantic Communication for Physical AI: Category Theory Meets Game Theory \vspace{-2mm}}
\vspace{-0mm}
\author{\IEEEauthorblockN{
Christo Kurisummoottil Thomas, \IEEEmembership{Senior Member, IEEE,} Walid Saad, \IEEEmembership{Fellow, IEEE,} \\ and Emilio Calvanese Strinati, \IEEEmembership{\normalsize Member, IEEE.}
\vspace{0mm}\vspace{-8mm}}
\thanks{\scriptsize Christo Kurisummoottil Thomas is with Department of Electrical and Computer Engineering, Worcester Polytechnic Institute, Worcester, MA, 01609, USA, Walid Saad is with the  Bradley Department of Electrical and Computer Engineering,
Virginia Tech, Alexandria, VA, 22305, USA, and Emilio Calvanese Strinati is with CEA-Leti, Grenoble, 38054, France. Emails:\{cthomas2@wpi.edu,walids@vt.edu, and emilio.calvanese-strinati@cea.fr\}. \\
\indent This research was supported by the U.S. National Science Foundation under Grant CNS-2225511. The work of Emilio Calvanese Strinati is supported by the EU-funded projects``6G-GOALS" under Grant agreement  No n.101139232 and "6G-DISAC"  under Grant Agreement No 101139130, and the French project funded by the program "PEPR Networks of the Future" of France 2030. }}
\maketitle

\vspace{-3mm}
\begin{abstract}\vspace{-0mm}
Physical artificial intelligence (AI) systems involve distributed sensing agents with embedded AI models that must coordinate efficiently to perceive, reason about, and act within networked environments in order to enable real-time operation. Traditional coordination approaches that transmit raw sensor data are limited by significant  communication overhead, processing latency, and redundant information transmission. While semantic communication (SC) can address these challenges by transmitting only critical task-relevant information,  existing deep learning-based joint source-channel coding methods exhibit limited adaptability to changing physical AI environments, poor generalization to out-of-distribution data, and scalability challenges. 
To overcome such challenges, in this paper, a comprehensive framework for \textit{compositional SC (CSC)} is proposed. CSC enables multiple heterogeneous physical AI sources to transmit semantic representations (SRs) that compose meaningfully at a base station (BS) or an edge server for remote inference tasks. In particular, a novel information-theoretic measure for compositional semantics is developed using category-theoretic constructs, that  quantifies the causal and compositional contribution of each device's concepts to inference tasks, capturing properties that mutual information cannot express. Second, Grothendieck topologies and presheaves are leveraged to formalize   semantic composition across heterogeneous devices, thereby ensuring that composed representations maintain consistency, context-invariance, and task relevance. 
Using the categorical frameworks as foundation, the multi-device coordination problem is formulated as a Stackelberg game in which physical AI devices (leaders) strategically commit to aligned encoding strategies and the BS (follower) optimally composes received SRs. An alternating direction method of multipliers (ADMM)-based algorithm is proposed to solve the game, enabling devices to compute equilibrium signaling strategies via limited message passing with the BS.  The existence of the equilibrium is shown under mild regularity conditions, and the equilibrium is Pareto optimal when compositional information yields increasing collective benefit. 
  Simulation results demonstrate that the proposed  game-theoretic approach achieves 
superior bandwidth reduction upto \SI{17}{\%}  over baselines, 
lower end-to-end latency,  \SI{53}{\%} less than 
cooperative multi-agent, distributed gradient descent, and uniform selection based CSC methods, while 
maintaining \SI{85}{\%} inference accuracy across diverse autonomous driving 
scenarios with multiple sensors and large number of semantic concepts.

\end{abstract}\vspace{-0mm}

\vspace{-7mm}
\section{Introduction}
\label{section_intro}
\vspace{-1mm}

Physical artificial intelligence (AI) systems comprise distributed, autonomous devices that can perceive, reason about, and act upon their environments in real-time \cite{mishra2025temporal}.  
Representative physical AI scenarios include autonomous robotic systems such as scalable heterogeneous fleets (humanoids, mobile robots, drones) coordinating toward shared industrial goals, autonomous vehicles navigating via shared situational awareness, and distributed teams enabling resilient coordination in disaster response. Other examples of physical AI include large-scale cyber-physical infrastructure such as smart cities that must jointly optimize transportation and energy,  and cyber-physical digital twins that need to maintain real-time understanding and autonomous decision-making across large-scale environments through cooperative inference. 
Different physical AI devices typically differ in sensing modalities and underlying AI architectures, leading to heterogeneous internal representations that encode the semantics of their environment. Because downstream decisions at a remote center, which is typically either an edge server or a base station (BS), often rely on information exchanged across devices, such misaligned representations can reduce inference accuracy, and hinder coordinated behavior. This creates a fundamental challenge for the physical AI ecosystem:
\emph{How can devices exchange compact, aligned representations that support accurate inference at a remote center in real-time with minimal bandwidth and latency overhead?}
In this regard,  transmitting raw sensor data incurs massive communication overhead from high-dimensional streams and introduces latency from centralized processing, while also wasting bandwidth due to redundancy across correlated observations. Alternatively, deploying deep learning (DL) models at the edge via federated learning (FL) \cite{Chen2020FederatedWireless} and \cite{Li2020Federated} can avoid transmission overhead, but it requires extensive retraining whenever environment dynamics change, or new sensor modalities or devices are introduced. This retraining  delays inference decisions beyond real-time constraints of physical AI systems and also prevents them from adapting when new autonomous agents or environmental conditions emerge. 
Moreover, FL operates on black-box AI that do not capture the semantics of the transmitted information and cannot support reasoning over the exchanged data. 
This necessitates mechanisms for coordinating representations that compose formally and interpretably while remaining \emph{task-relevant, semantically consistent, and generalizable} across changing sensor configurations and environment dynamics.

To address these limitations, semantic communications (SC) \cite{Cstrinati20216g , KountourisCommMag2021,ChaccourArxiv2022,Lu2022Semantic} emerged as a promising paradigm that minimizes transmitted data by exploiting reasoning techniques to reconstruct or generate high-quality information with essential semantics at the receiver (RX), rather than transmitting raw observations. 
However, existing SC works such as \cite{Cstrinati20216g , KountourisCommMag2021,ChaccourArxiv2022,Lu2022Semantic}, which rely on DL-based joint source-channel coding (JSCC), have several key limitations. First, deep JSCC has \emph{limited real-time adaptability}, since data-driven DL models  typically require extensive retraining when distribution shifts occur, making them unsuitable for real-time applications where environmental conditions continuously evolve. 
Second, the large scale deployment of existing deep JSCC methods are hindered by \emph{scalability challenges}, as the number of data sources grow and their modality heterogeneity increases. 
 Third, data at distinct sources may be processed by different AI models, so the same real-world concept can be encoded as different semantic representations (SRs), where an SR is a compact learned vector capturing the task-relevant content of a device's observation. Because these SRs live in incomparable, device-specific spaces, the remote center must align them before they can be combined \cite{HuttebrauckerArxiv2024,fiorellino2025frame,sana2023semantic,pannacci2025semantic,huttebraucker2024soft,grimaldi2025learning}. Fusing them without such alignment corrupts the meaning of the composed result and degrades the quality of the downstream inference it supports.

\vspace{-4mm}\subsection{Prior Works}\vspace{-1mm}

The state-of-art in SC ~\cite{ShaoArxiv2022,Lu2022Semantic,YanWCL2022,WangJSAC2022,XieJSAC2022,Li2023NOMA,seo2023semantic,Nguyen2024SwinSemantic} largely focuses on data-driven DL–based semantic encoders and decoders that compress data. For instance, the   works in \cite{ShaoArxiv2022} and \cite{Lu2022Semantic,YanWCL2022,WangJSAC2022} advance the basic SC paradigm through conceptual frameworks, resource allocation, and optimization methods, but they assume single-device or single-modality settings with end-to-end trained models. Multiuser SC has been explored in \cite{XieJSAC2022,Li2023NOMA,seo2023semantic,Nguyen2024SwinSemantic}. Specifically, the authors in \cite{XieJSAC2022} introduce task-oriented multi-user SC, where multiple devices transmit semantically aligned information for a shared task, while \cite{Li2023NOMA}  develop an SC solution that improves spectrum efficiency and     multi-user coordination via non-orthogonal multiple access  without degrading task performance. Unlike these physical-layer–focused works \cite{XieJSAC2022,Li2023NOMA,seo2023semantic}, the work in \cite{Nguyen2024SwinSemantic} applies task-oriented multi-user SC to visual question answering, demonstrating  that collaboratively transmitted semantic features improve inference quality over single-device approaches. Existing works on SC for remote inference include \cite{ChenTCCN2023} and \cite{LoWCL2023}. In \cite{ChenTCCN2023}, the authors use spiking neural networks for sensing, processing, and communication. Such approaches are energy-efficient, however, they yield low inference accuracy, and are accompanied with ineffective training for continuous-valued data. Moreover, the solution of \cite{ChenTCCN2023} has limited applicability to practical scenarios.  The work in \cite{LoWCL2023} focuses on composing multi-view images of a single object or event to improve inference accuracy, but remains data-driven and lacks generalizability and scalability.

Prior reasoning-driven point-to-point SC approaches \cite{ChristoTWCArxiv2022,ChristoJSAITArxiv2023,dacunto2025causal,d2026PdLlearning} use causal reasoning and neurosymbolic AI, but face a fundamental scalability barrier in multi-device systems with heterogeneous modalities. As the number and diversity of sensing sources grow, aligning independently learned SRs and composing them meaningfully becomes harder. 
Herein, we must transition to compositional SC (CSC), that enables combining distinct semantics from diverse sensing sources and then perform deductive reasoning to generate novel information to aid in several downstream tasks. 
Compositionality complements causality in an important way. While causal models are typically constrained to inference over fixed graph structures, compositional approaches enable reassembly of semantic concepts or even distinct causal graphs into novel configurations. This supports generalization to scenarios and sensor combinations that purely causal models cannot handle on their own.

\vspace{-5mm}\subsection{Contributions}
\vspace{-1mm}
The main contribution of this paper is a comprehensive framework for CSC, grounded in the principle that compositionality \cite{SinhaArxiv2024} and \cite{ boeshertz2025predictive} , which is the ability to decompose complex structures into simpler components and recombine them into novel configurations, enables physical AI systems to generalize beyond their training data. Unlike existing SC approaches that transmit task-specific feature encodings, we formalize compositionality via category-theoretic structures and game-theoretic mechanisms. This enables heterogeneous physical AI devices, such as robotic platforms with diverse sensors and AI models, to coordinate through aligned SRs in which concepts extracted by one device compose with those from others. Consequently, devices coordinate without task-specific retraining, novel sensor combinations generalize to new environments, and communication scales with conceptual complexity rather than raw data, reducing bandwidth and latency. As such, our key contributions include:
\vspace{-1mm}\begin{itemize}
    \item We develop a novel information-theoretic measure for compositional semantics using category theory constructs such as fibrations \cite{velarde2024fibration} and Grothendieck topos \cite{Grothendieck1972} and \cite{BelfioreBennequin2021Topos}. Unlike mutual information, which captures only statistical dependence, our measure is grounded in intrinsic information and quantifies the causal and compositional contribution of each device's concepts to downstream tasks. This enables principled decisions about what semantic content to transmit.
    \item We use category-theoretic structures such as Grothendieck topology and presheaves to formalize how SRs from heterogeneous devices can be composed while maintaining semantic consistency and preserving task relevance.
    \item We formulate the SR alignment and composition problem as a noncooperative Stackelberg game in which devices (leaders) commit to encoding strategies and the BS (follower) observes all representations before optimally composing them. We derive the equilibrium composition and semantic encoding strategies as solutions to this game, providing a theoretically grounded protocol for decentralized SC.
    \item We analytically prove that: (i) A Stackelberg equilibrium (SE) exists under mild regularity conditions on strategy spaces and utility functions (Theorem 2),  and (ii) the equilibrium achieves Pareto efficiency when device concepts are synergistic (Theorem 4), providing theoretical foundations for practical CSC implementation.
    \item We develop a practical algorithm based on alternating direction method of multipliers (ADMM) that enables devices to compute equilibrium strategies through limited message passing with the BS, proving convergence to the SE while maintaining full decentralization and scalability.
    
    \item Simulation results on multi-sensor inference tasks demonstrate that proposed CSC maintains high inference accuracy with upto $17\%$ reduction in communication overhead and $53\%$ lower end-to-end latency compared to baseline data-driven SC methods, while maintaining robust generalization to novel unseen scenarios.
\vspace{-1mm}\end{itemize}
The rest of the paper is organized as follows. Section II presents the compositional model. Section III develops a category-theoretic semantic information measure. Section IV formulates the Stackelberg game and derives equilibrium solutions via backward induction. Section V presents simulation results, and Section VI concludes the paper.

\begin{figure*}[t]
\centerline{\includegraphics[width=0.77\textwidth]{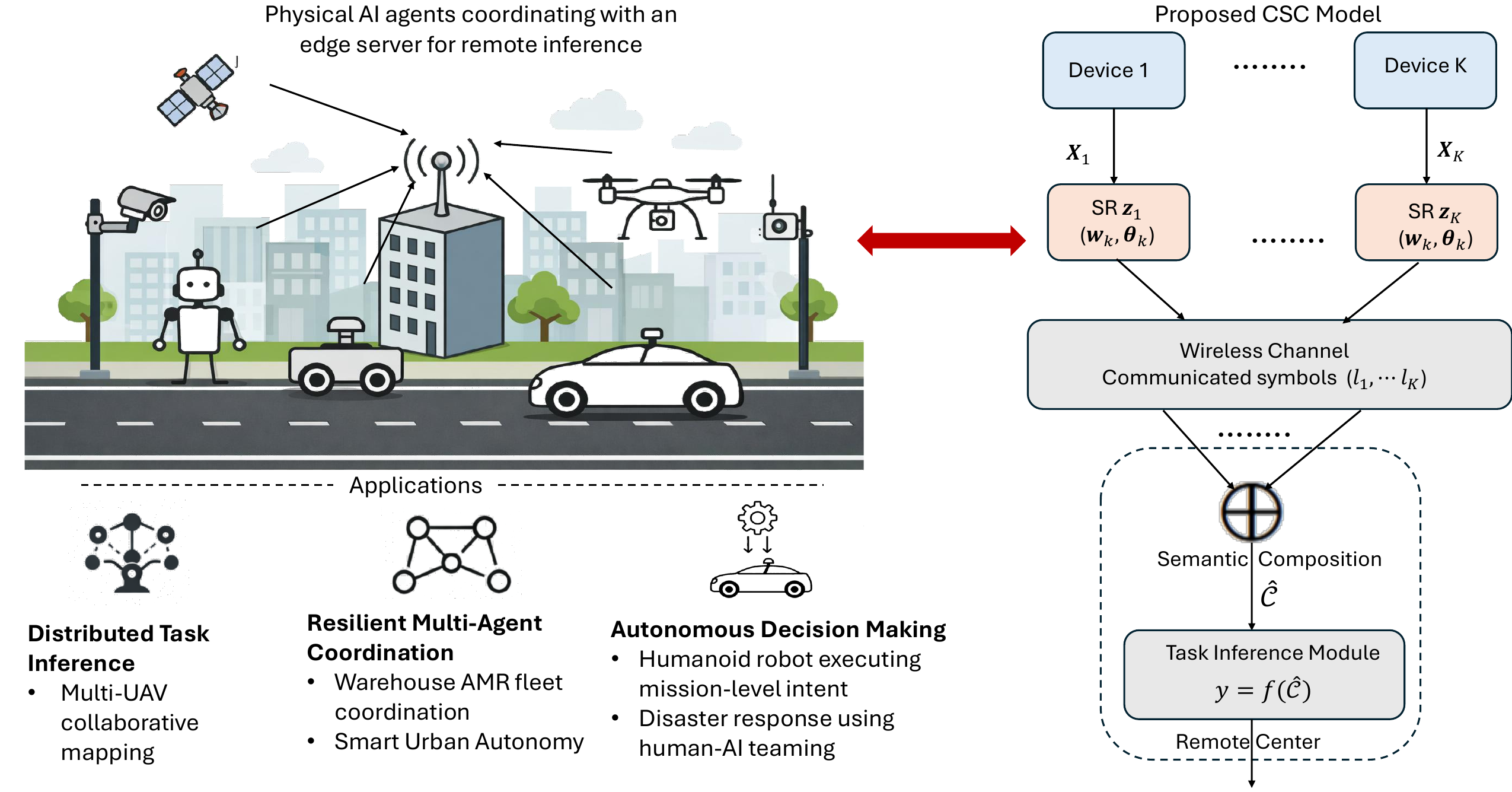}}\vspace{-3mm}
\caption{\small Overview of the proposed CSC system, with applications across physical AI and networking.}
\label{NetworkDiagram}\vspace{-0mm}
\vspace{-5mm}
\end{figure*}
\vspace{-3mm}\section{System Model}\vspace{-1mm}

\subsection{System Architecture and Problem Setup}\vspace{-1mm}

Consider a physical AI system illustrated in Fig. 1, composed of a set $\mathcal{K}$ of $K$ autonomous agents or robotic platforms, each equipped with heterogeneous sensors for environmental perception. 
Each device $k \in \mathcal{K}$  is equipped with sensors that observe local data, represented using the matrix $\boldsymbol{X}_k \in \mX \subseteq\mathbb{R}^{d \times N}$, where $d$ is the feature dimension and $N$ is the number of raw features. 
Instead of transmitting the entire raw data features $\boldsymbol{X}_k,  \forall k \in \mathcal{K}$, which incurs an overhead of order $O(K d \times N)$ bits, latency from transmission and centralized processing, and redundancy due to correlated observations across devices, we will adopt an SC approach. In this approach, each device $k$ extracts semantic concepts $\mC_k$ from $X_k$. Here, \emph{semantic concepts}\footnote{\scriptsize The fundamental challenge we address is not how to extract concepts from raw data at the device or received data at the BS, but rather how to align and coordinate SRs across heterogeneous devices such that independently-learned concepts compose meaningfully at the RX. Thus, we assume that concept extraction is a solved problem at both the transmitter (TX) and  RX ends, can be efficiently handled using existing methods, such as energy-based models \cite{du2021unsupervised}, thereby focusing exclusively on the multiuser language coordination problem: ensuring that device $k$'s SR $\bmz_k$ can be correctly interpreted and composed with representations from other devices in a way that preserves task relevance and compositional structure. } are those features present in the data that have relevance with respect to the remote inference tasks at the BS. Each device transmits SRs $\bmz_k \in \mathbb{R}^m$, learned using our approach in Section~\ref{eq_language_alignment}, and consisting of   minimalistic representation of the semantic concepts, to the BS, where $m \ll d \times N$.  The BS then receives the set $\{\bmz_1, \dots, \bmz_K\}$ and extracts the concepts $\mC_i$.  Specifically, the BS aligns and combines the semantic information from all devices into a single composed representation:
$\hat{\mathcal{C}} = \bigoplus_{i=1}^K \hat{\mathcal{C}}_i,$
where the composition operation $\bigoplus$
 integrates the semantic content from each device $
i$ into a cohesive whole. The BS then performs \emph{remote inference}, which is defined as the downstream task that depends on multiple devices' perspectives, by applying a task-specific function to the composed representation:
$
\mathbf{y} = f(\hat{\mathcal{C}}),$
where ${\mathbf{y}}
$ is the inference result. This result drives the \emph{application} layer, which takes actions based on the collective understanding such as issuing collision-avoidance maneuvers, generating anomaly alerts, or making physical AI control decisions.
The key challenge is that devices $\mK$ learn their representations $\bmz_k$ independently using different sensor modalities and neural networks, resulting in incomparable representation spaces $E_1, \ldots, E_K$. The composition operation $\bigoplus$
 must therefore not only combine information from all devices but also ensure semantic alignment, that a concept encoded by one device is correctly interpreted by the BS and can be meaningfully combined with concepts from other devices. This is the core problem that CSC addresses.

\vspace{-4mm}\subsection{Compositional SRs: Basic Definitions}\vspace{-1mm}

We now formalize how concepts should be represented to enable compositional reasoning.
\vspace{-2mm}\begin{definition} \label{CR} A \emph{semantic concept representation function} $R : \mC\rightarrow \mR^d$ is compositional if for any concepts $c_i, c_j \in \mC$ and their composition $ c_i \oplus c_j$, there exist weights $w_i, w_j \in \mathbb{R}_+$ such that:
$
    R(c_i \oplus c_j) = w_i R(c_i) + w_j R(c_j) .$
\vspace{-2mm}\end{definition}
In other words, the representation of composed concepts is a weighted sum of individual concept representations. The weights reflect the relative semantic importance of each concept to the composed object, and depend on remote inference  at the BS. Compositional representations enable deductive reasoning such that the BS can combine individual concept vectors to understand complex objects. For example, if device 1 sends \{person, sitting\} and device 2 sends \{conversation\}, the BS can compose these to infer \{two people, sitting, conversing\}. The composition example above assumes that each device's representation $\bmz_k$ encodes its underlying concepts in a separable manner. However, SRs can be information-rich while remaining conceptually entangled, making composition unreliable at the BS. To ensure that transmitted representations preserve compositional structure, we introduce a \emph{compositionality measure}, a metric that quantifies how faithfully a SR captures the semantics associated with its underlying concepts. This measure will serve as a design objective in our multiuser coordination framework in Section~\ref{eq_language_alignment}, incentivizing devices to transmit representations suitable for compositional reasoning. 
\vspace{-9.5mm}\subsection{Compositionality Measure}\vspace{-1mm}
For any SR $\bmz_k$, we quantify how faithfully it represents its underlying concepts using cosine similarity:
\vspace{-2mm}\begin{equation}
C(\bmz_k)= \frac{1}{D} \sum_{d=1}^{D} \left| \cos(\bmz_k, R(c_{k,d})) \right|.
\label{eq_cs_score}
\vspace{-2mm}\end{equation}
 We use cosine similarity in \eqref{eq_cs_score} because compositional SRs should be expressible as weighted combinations of concept vectors (Definition~\ref{CR}), and cosine similarity directly measures angular alignment regardless of magnitude. A high compositionality score ($\mathcal{C}(\bmz_k) \to 1$) indicates that $\bmz_k$ preserves clear compositional structure, enabling the BS to decompose it into constituent concepts for deductive reasoning. Conversely, a low score means that $\bmz_k$ entangles concepts in ways that complicate semantic interpretation and composition at the RX.

Composing semantics from heterogeneous devices raises two challenges that the rest of the paper addresses. First, we must rigorously quantify the task-relevant, compositional information each SR carries, which we develop in Section III. Second, we must coordinate devices whose contributions are interdependent, since each SR's value depends on how it complements others and on the BS's composition strategy, which we cast as a strategic game in Section IV.

\vspace{-4mm}\section{Compositional Semantic Information using Category Theory}\vspace{-1mm}
\label{CompositionalSemanticInformation}\vspace{-1mm}

Standard measures such as entropy and mutual information capture statistical dependence but not the compositional value an SR carries for downstream inference. We therefore develop a measure that jointly captures \emph{task relevance, compositionality, and semantic alignment}, built on category-theoretic constructs. Lenses \cite{steckermeier2015lenses} formalize the bidirectional encode-decode transformation with consistency guarantees under noisy channels, and fibrations  \cite{velarde2024fibration} let heterogeneous device semantics compose over a shared communication language without requiring identical representations.

\vspace{-5mm}\subsection{Lenses: Modeling Bidirectional Semantic Transformations}\vspace{-1mm}

In SC, two operations are required. The first operation, $g$, extracts semantics from data, which corresponds to encoding. The second operation, $p$, updates or generates data based on semantic extraction, which corresponds to decoding or reconstruction. These operations must be consistent, meaning that performing the get operation followed by the put operation should yield the original semantics, and vice versa. This motivates the use of lenses in category theory:
\vspace{-3mm}\begin{definition}\label{lens}
A \emph{lens} $L = (\mX, \mC, g, p, \alpha, \beta)$ consists of: (a)  the \emph{syntax category} $\mX$, representing the data or observation space; (b) the \emph{semantic category} $\mC$,  representing the semantic concept space; (c) the ``get" operation, a  mapping: $g: \mX \to \mC$ that performs semantic extraction; and (d) the ``put" operation, a mapping:  $p: \mX \times \mC \to \mX$ that performs semantic-guided data generation at the RX. For the put operation, the input consists of the current data element $x \in \mX$ and the target semantics $b' \in \mC$. The output is a new data element $x'$ that realizes the semantics $b'$. A key property of the lens is that $g(p(x, b')) = b'$, which ensures that the generated data has the desired semantics. (e) The magma structure $\alpha: B \otimes B \to B$ defines concept composition, and (f) the comagma structure $\beta: B \to B \otimes B$ defines concept decomposition.
\vspace{-3mm}\end{definition}\vspace{-0mm}

Definition~\ref{lens} formalizes encoding (get)  and decoding (put) with a consistency guarantee: $g(p(x, b')) = b'$, which means semantics are preserved in the encode-decode cycle. The composition ($\alpha$) and decomposition ($\beta$) operations enable devices to combine semantic concepts meaningfully at the RX without semantic loss.
However, errors in physical wireless channel and inconistencies in semantic understanding require defining a \emph{weak lens laws}. Weak lens law relaxes the strict lens laws by allowing approximate consistency. The first relaxed law,  $(L_1^\prime)$, states that $g(p(x, b)) = b$, which ensures that the put operation generates the correct semantics. The second relaxed law, $(L_2^\prime)$, requires that the semantic distance between $g(x)$ and $g(p(x, g(x)))$ is within a tolerance, expressed as $d_B(g(x), g(p(x, g(x)))) \leq \varepsilon$, where $d_B$ is a semantic distance metric in the semantic category $B$. This law enforces weak semantic consistency. The third relaxed law,  $(L3')$, states that $p(p(x, b_1), b_2) = p(x, b_2)$, which means that the last update takes precedence over previous updates. Next, we  illustrate via an example how get and put functions can be used.
\vspace{-1mm}\subsubsection{Scenario A- alter the semantic attribute of an image} Consider a reference image $x_{\text{ref}}$ that shows a red car. The SR of $x_{\text{ref}}$ is given by $\text{get}_k(x_{\text{ref}}) = \{\text{color: red}, \text{object: car}, \text{background: street}\}$. By changing the color to blue, we define a new semantic concept $c_{\text{new}} = \{\text{color: blue}, \text{object: car}, \text{background: street}\}$. The updated image is obtained by applying the put operation: $x_{\text{new}} \!=\! \text{put}_k(x_{\text{ref}}, c_{\text{new}})$, which results in an image of a blue car on the same street. Although $x_{\text{new}} \!\neq\! x_{\text{ref}}$, the SR satisfies $\text{get}_k(x_{\text{new}})\! = \!c_{\text{new}}$, preserving the structural composition of $x_{\text{ref}}$.

\subsubsection{Scenario B: Perform multi-device composition at BS}
In this example, a BS receives concepts from multiple devices and generates a unified representation. Device 1, which uses a camera, provides the concept $z_1 = \{\text{person, sitting}\}$. Device 2, which uses audio input, provides the concept $z_2 = \{\text{conversation, two people}\}$. Device 3, which uses motion sensing, provides the concept $z_3 = \{\text{stationary}\}$. The BS composes these semantics using the composition function $\alpha$, resulting in $c_{\text{composed}} = \eta(\bmz_1, \bmz_2, \bmz_3) = \{\text{two people, sitting, conversing, stationary}\}$.
Given a lens structure for a compositional concept, we next look at how to formalize $\eta$ of a lens structure for the composition of multiple concepts. Such a composed lens structure will be crucial to defining the concept of compositional semantic information.
\vspace{-2mm}\begin{lemma}
    Given a family of compositional concepts $\mC_1,\cdots,\mC_n, n \in \mN$, their combination defined by $\eta: \mC_1 \otimes \cdots \otimes \mC_n \rightarrow \mC_1 \otimes \cdots \otimes \mC_n$, forms a new lens structure: $\left(\bigcirc_{i=1}^n \mX_i,\bigcirc_{i=1}^n \mC_i,\bigcirc_{i=1}^np_i,\bigcirc_{i=1}^ng_i,\bigcirc_{i=1}^n\alpha_i,\bigcirc_{i=1}^n\beta_i \right)$.
\end{lemma}
\vspace{-1mm}\begin{IEEEproof}[Proof Sketch]
    When two lens structures 
\((\mathcal{X}_1, \mathcal{C}_1, p_1, g_1, \alpha_1, \beta_1)\) 
and 
\((\mathcal{X}_2, \mathcal{C}_2, p_2, g_2, \alpha_2, \beta_2)\) 
have commuting operations 
\((p_1 \circ p_2 = p_2 \circ p_1 \text{ and } g_1 \circ g_2 = g_2 \circ g_1)\), 
their composition into
$
(\mathcal{X}_1 \circ \mathcal{X}_2,\;
\mathcal{C}_1 \circ \mathcal{C}_2,\;
p_1 \circ p_2,\;
g_1 \circ g_2,\;
\alpha_1 \circ \alpha_2,\;
\beta_1 \circ \beta_2)
$
yields a new lens structure that also satisfies Definition~1 
(proof in \cite{HeffordArxiv}). Therefore,  
by induction, any finite composition 
\(\bigcirc_{i=1}^n\) of lens structures with pairwise commuting operations 
forms a lens structure, enabling hierarchical concept composition across multiple devices.
\end{IEEEproof}

Using Lemma 1, we have established that concepts hierarchically compose through correlators $\eta$. However, organizing these compositional rules into a unified categorical framework requires more than correlator definitions alone. The Karoubi envelope provides this framework by constructing a category where concept spaces $\mC$ are paired with their correlators $\eta$, and morphisms between concepts respect the compositional constraints of both spaces. This categorical structure enables \emph{local compositionality} within each device's concept hierarchy, \emph{morphism preservation} ensuring compositions are consistent across concept transitions, and  \emph{hierarchical coherence} allowing compositions at one level (device) to influence compositions at higher levels (remote inference at BS).
\vspace{-1mm}\begin{definition}\label{Karoubi}\vspace{-2mm}

The \emph{Karoubi envelope} $\bar{\mathcal{C}}$ pairs each concept with a correlator (a constraint function):
$(\mC, \eta)$ in $\bar{\mathcal{C}}$ represents a semantic concept, where $\mC$ is base concept space and $\eta$ is a correlator defining relations among concepts. The morphisms $f : (\mC_1,\eta)\rightarrow (\mC_1,\sigma)$ are the morphisms $f:\mC_1 \rightarrow C_2$ such that $\sigma \circ f = f = f\circ \eta$.
\vspace{-2mm}\end{definition}
Definition~\ref{Karoubi} means that the Karoubi envelope pairs each concept with a correlator, which is a rule that defines how concepts hierarchically compose to form higher-level semantic structures for downstream inference. For example, basic concepts from camera (``vehicle ahead") and LIDAR (``distance$=5$~m", ``closing speed") do not directly produce an inference; instead, they must compose hierarchically: \{\text{``vehicle ahead", ``distance}$=5$~m", \text{``closing speed"}\} $\rightarrow$ {\text{``approaching vehicle danger"}} $\rightarrow$ \text{inference ``activate collision avoidance."} The correlator $\eta$  encodes these compositional relationships, and the projection function $p(c) = \{c' \in \mC \mid \eta(c) = \eta(c')\}$ groups concepts that participate in the same hierarchical composition pathways. This is essential in multi-device SC: when camera and LIDAR transmit different basic concepts, the correlator ensures they are correctly combined through compositional pathways that lead to meaningful, task-relevant inferences. 

Having formalized compositional hierarchies via the Karoubi envelope and its correlators, we now introduce a novel information-theoretic measure, $\mathbb{S}(\bmz)$, that captures compositional semantic information by accounting for both individual concept contributions and their composition through pathways defined by the Karoubi correlator.
\vspace{-4mm}\subsection{Compositional Semantic Information}\vspace{-1mm}
Let  $\bar{\mC}$ be the Karoubi envelope and $\mathbb{P}(\mC)$ the power set for   a set of concepts $\mC$, with projection map $p:\bar{\mC}\to \mathbb{P}(\mC)$. We define task-relevant compositional objects as:
\vspace{-3mm}\begin{equation}
\hat{\mathcal{P}}(\mC) = \{ B \subseteq \mC : \exists c \in \bar{\mC}, \; p(c) = B, \; \mathbb{S}(c; B) > 0 \},\vspace{-3mm}
\end{equation}
where $\mathbb{S}(c; B)$ is the intrinsic information (see Definition 4) measuring the contribution of concept $c$ to compositional object $B$. This restriction ensures we only consider compositions that form valid projections in the Karoubi envelope, and
carry non-zero information about the inference task.
\vspace{-3mm}\begin{definition}
 
The \emph{compositional semantic information} conveyed by a representation $\bmz$ is
\vspace{-2mm}\begin{equation}
\mathbb{S}(\bmz)
=
\sum_{B \in \mathbb{P}(\mC)}
\sum_{c \in \bar{\mC}}
P(c \mid \bmz)\,\mathbb{S}(c; B)\,\mathds{1}\!\left(p(c)=B\right).\label{eq_compositional_info}\vspace{-2mm}
\end{equation}
\end{definition}
Here $\mathcal{P}(\mathcal{C})$ is the power set of $\mathcal{C}$, enumerating all feasible compositional objects $B$; each object in $\bar{\mathcal{C}}$ pairs a concept $c$ with a correlator specifying how its constituents compose. The posterior $P(c \mid \bm{z}) = P(\bm{z} \mid c) P(c) / P(\bm{z})$ weights each concept by its likelihood under $\bm{z}$, capturing that $\bm{z}$ may encode concepts probabilistically. The indicator $\mathbf{1}(p(c) = B)$ retains only concepts whose categorical projection matches $B$, enforcing consistency between projection and semantic composition. The term $S(c; B)$ is the intrinsic information that $c$ contributes to forming $B$ when composed with other concepts, defined as\footnote{In practice, $P(c \mid \bm{z})$ and $p(B \mid c, c', \pi)$ are approximated by learned energy functions $E_\psi$ and $E_\phi$; see [34].}
\vspace{-2mm}{\beq
\begin{aligned}\small
 &\mathbb{S}(c; B) =\\ & \sum_{\pi \in \Pi} P(\pi) \!\!\!\sum_{c' \in \mC\setminus c} P(c' \mid c, \pi) \big[ - \ln p(B \mid c, c', \pi) \big] p(B \mid c, c', \pi),
\end{aligned}
\vspace{-2mm}\eeq}
where,
$\Pi$ is the set of possible correlators (composition rules), with $\sum\limits_{\pi \in \Pi} P(\pi) = 1$. Each correlator $\pi$ encodes a distinct pathway for how $c$ and other concepts hierarchically combine to form $B$, $P(c' \mid c, \pi)$ is the conditional probability of co-concept $c'$ appearing alongside $c$ under correlator $\pi$. This averages over all possible co-concept configurations compatible with the compositional structure $\pi$ defines, $p(B \mid c, c', \pi)$ is the conditional probability that compositional object $B$ is formed when concepts $c$ and $c'$ compose under correlator $\pi$, The inner summation over $c' \in \mC\setminus c$ marginalizes over all possible co-concept configurations, capturing how $c$'s contribution changes depending on which other concepts it composes with. $\mathbb{S}(c; B)$ quantifies the expected information gain from knowing concept $c$ when forming $B$, averaged over all possible correlators (hierarchical pathways) and all possible synergistic co-concepts. High intrinsic information means $c$ is consistently informative for B across different composition pathways; low intrinsic information means $c$'s contribution is context-dependent or redundant.

The compositional semantic information conveyed by $\bmz$ about $\bmy$ can be obtained as follows.
\vspace{-2mm}\beq
\begin{aligned}
&\mathbb{S}(\bmy; \bmz) = \mathbb{E}_{p(\bmzh\mid\bmz)}\mathbb{S}(\bmzh\mid \bmz) \\= & \mathbb{E}_{p(\bmzh\mid\bmz)}\left[\sum\limits_{\widehat{c}} p(\widehat{c}\mid \bmzh) \sum\limits_{\bmy \in \mathbb{P}(\mC)} \sum\limits_{c\in \bar{\mC}}\mathbb{S}(\widehat{c};\bmy) p (p(\widehat{c})=\bmy)\right]  \\=
& \underbrace{\mathbb{E}_{p(\bmzh\mid\bmz)}\left[\sum\limits_{\widehat{c}} p(\widehat{c}\mid \bmzh)  \sum\limits_{c\in \bar{\mC}}\mathbb{S}(\widehat{c};\bmy) p (p(\widehat{c})=\bmy)\right]}_{\mbox{Task relevant information}}  \\ +&\underbrace{\mathbb{E}_{p(\bmzh\mid\bmz)}\left[\sum\limits_{\widehat{c}} p(\widehat{c}\mid \bmzh) \sum\limits_{\bmy^{\prime} \in \mathbb{P}(\mC) \setminus \bmy} \sum\limits_{c\in \bar{\mC}}\mathbb{S}(\widehat{c};\bmy^{\prime}) p (p(\widehat{c})=\bmy^{\prime})\right]}_{\mbox{                Generalization potential (off-task concepts)
}}
\label{eq_Sy_z}
\end{aligned}
\vspace{-2mm}\eeq
The first term in \eqref{eq_Sy_z} represents the information contained in $z$ that is directly relevant to inferring $y$. The second term represents the information that enables generalization to other tasks captured by $y^{\prime}$. This decomposition plays a crucial role in multi-task learning for physical AI systems.

\vspace{-4mm}\subsection{Grothendieck Topology: Ensuring Semantic Consistency
}\vspace{-1mm}
Next, we establish  consistency guarantees at the representation level. When device $k$ transmits representation $\bmz_k \in \mathbb{R}^m$, it encodes its local concepts $\mathcal{C}_k$ using a device-specific semantic space. The BS receives $K$ heterogeneous representations ${\bmz_1, \cdots, \bmz_K}$ from distinct semantic spaces and must compose them to form compositional objects B that contribute to inference task $\bmy$. This challenge of \emph{semantic consistency in the representation space} 
can be addressed by imposing a Grothendieck topology on the representation space.

\vspace{-2mm}\begin{definition}
Let $\mathcal{C}_{\text{repr}}$ be a category whose objects are device semantic spaces $\mZ_1, \mZ_2, \dots, \mZ_K$ and the BS inference target space $\mY$, and morphisms are $f: \mZ_k \to \mZ_j$ that represent semantic alignment/consistency maps showing that concepts from device $k$ can be reliably interpreted in device $j$'s semantic framework.
A morphism $f: \mZ_k \to \mY$ indicates that representation $z_k \in \mZ_k$ can directly contribute to task $y \in \mY$.
A \emph{presieve} $P$ on $\mY$ is a collection of morphisms $\{ f_k: \mZ_k \to \mY \mid k \in \{1, \dots, K\} \}$ representing all devices that could potentially contribute to task $y$.
\vspace{-2mm}\end{definition}
For example, if $c$ represents  concept ``vehicle”, then the presieve $P$ may include morphisms such as $\{\text{car} \to \text{vehicle}, \text{truck} \to \text{vehicle}, \text{bicycle} \to \text{vehicle}, \dots\}$. 
However, to capture compositional concepts, we  use the concept of sieve of task-relevant representations, defined next.
\vspace{-2mm}\begin{definition}
A \emph{sieve} $S_{y}$ on $\mY$ is a presieve satisfying closure under composition. This means that if $f: \mZ_k \to \mY$ is in $S_{y}$, and $g: \mZ_j \to \mZ_k$ is a consistency morphism, then $f \circ g: \mZ_j \to \mY$ is in $S_{y}$.
\vspace{-1mm}\end{definition}
If device $k$'s representation $\bmz_k$ can contribute to $y$, and device $j$'s representation $z_j$ is semantically consistent with $\bmz_k$ (morphism $g$ exists), then $z_j$ can also contribute to $y$ through composition.
Further, we define the concept of Grothendieck topos \cite{Grothendieck1972} that provides a framework to capture all the compositional semantics represented by any concept, either single or a composed $c$.

\vspace{-2mm}
\begin{definition}
\label{Grothendieck_topos}
\emph{A Grothendieck topology} on a category $\mathcal{C}$ is a rule $J$ that assigns to each object $c \in \mathcal{C}$ a collection $J(c)$ of sieves on $c$ such that:
\begin{itemize}
    \item \textbf{(Maximality Axiom)} The maximal sieve 
    $M_c = \{ f : d \to c \mid d \in \mathrm{Ob}(\mathcal{C}) \}$ 
    is in $J(c)$.

    \item \textbf{(Stability Axiom)} If $S \in J(c)$ and $f : d \to c$, then the pullback sieve
    $
    f^*(S) = \{ g : e \to d \mid f \circ g \in S \}
    $
    is in $J(d)$.

    \item \textbf{(Transitivity Axiom)} If $S \in J(c)$ and $R$ is a sieve on $c$ such that
    $
    f^*(R) \in J(d) \quad \text{for all } f : d \to c \text{ in } S,
    $
    then $R \in J(c)$.
\end{itemize}
\vspace{-3mm}\end{definition}
This means that a Grothendieck topology ensures that semantic interpretations are complete, consistent under context changes, and compositionally closed. 
\vspace{-2mm}\begin{theorem}
If transmitted representations $\{\bmz_k\}_{k \in \mathcal{K}}$ satisfy the Grothendieck topology consistency conditions of Definition~\ref{Grothendieck_topos}  on representation space $\mathcal{C}_{\text{repr}}$, then:
\begin{enumerate}
    \item \textbf{Consistency:} The composed representation $\bmz = \bigoplus_{k}\bmz_k$ is uniquely interpretable at the BS (no semantic ambiguity).
    \item \textbf{Robustness to device dynamics:} If a device $k$ leaves the system, the sieve $S_y$ shrinks but remains well-defined (Stability Axiom), so the remaining devices' representations still compose consistently.
    \item \textbf{Multi-hop task relevance:} If $\bmz_k \to \bmz_{\text{composed}} \to \bmy$ ($\bmz_k$ contributes through intermediate composition), then $\bmz_k$ is task-relevant (Transitivity Axiom).
\end{enumerate}
\end{theorem}
\vspace{-2mm}\begin{IEEEproof}
 See Appendix~\ref{Proof_Theorem1}.
\end{IEEEproof}
The Grothendieck topology at the representation level is induced by Karoubi correlators at the concept level. Specifically, consistency morphisms $g: \mathcal{Z}_j \rightarrow \mathcal{Z}_k$
 exist if and only if correlators $\eta_j$
and $\eta_k$
 are compatible—meaning they decompose representations into alignable concept hierarchies. By Theorem 1, this ensures that compositional semantics at the concept level (Karoubi envelope) induce semantic consistency at the representation level (Grothendieck topology). However, the Grothendieck framework assumes these morphisms already exist and does not specify how to construct them when device representations are initially incomparable. We address this gap by introducing fibrational encoders, which explicitly construct morphisms by mapping device-specific representations to a shared communication vocabulary $\mathcal{L}$
.

\vspace{-4mm}\subsection{Fibrational Communication Language}\vspace{-1mm}

\begin{figure}[t]
\vspace{-4mm}\centerline{\includegraphics[width=0.8\columnwidth]{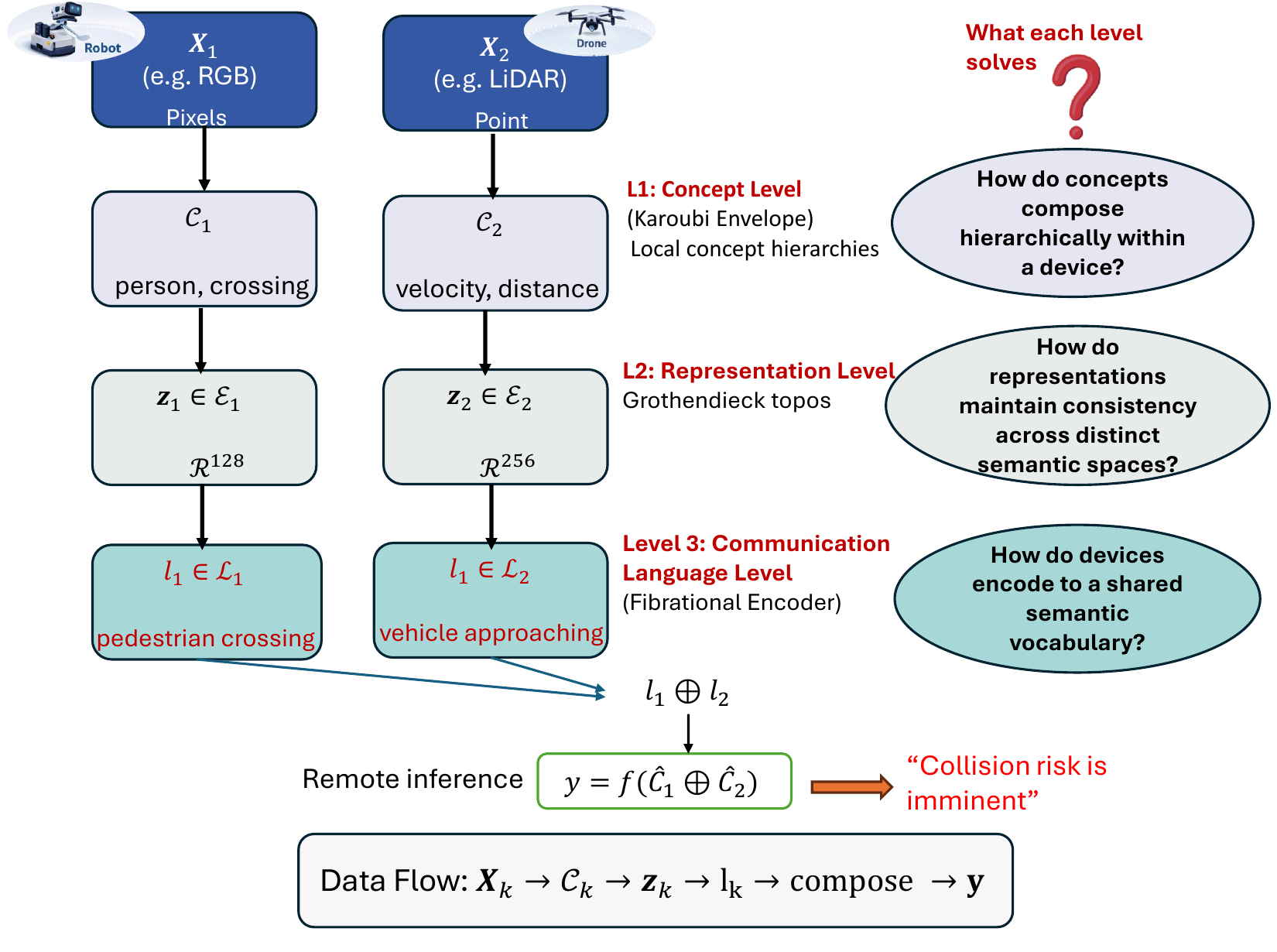}}\vspace{-3mm}
\caption{\small Hierarchical semantics via Fibrational semantic language.}
\label{BlockDiagram}\vspace{-2mm}
\vspace{-0mm}
\end{figure}\vspace{-0mm}

The CSC framework operates across three conceptually distinct layers, each operating on different data structures. The transformation from raw data to inference follows the pipeline: $\bmX_k$ (raw observations) $\rightarrow \mC_k$ (extracted concepts) $\rightarrow\bmz_k$ (representation in latent space $\mZ_k$) $\rightarrow$ $\ell_k$ (semantic symbol in $\mL$). At the \emph{concept level}, the Karoubi envelope  formalizes how semantic concepts hierarchically compose within each device's local semantic space through paired structures of concepts and correlators that define compositional pathways. At the \emph{representation level}, the Grothendieck topology  formalizes how SRs from heterogeneous devices maintain consistency despite living in different latent spaces, ensuring that when representations compose at the BS they do so without semantic ambiguity. In this section, we focus on the \emph{communication language level}, and use fibrations to formalize how heterogeneous local representations are encoded into a shared communication vocabulary that enables devices to transmit information interpretable by the BS without requiring device-specific alignment protocols.

A communication language is defined as a category $\mL$ where objects are  semantic symbols $\ell_1,\ell_2,\cdots,\ell_M$, representing unique and non-redundant concepts such as person, sitting, approaching, and morphisms $\ell_i\rightarrow \ell_j$ encode semantic relationships (e.g., person $\rightarrow$ entity or approaching $\rightarrow$ moving). The communication language functions as a universal semantic vocabulary shared across all devices, where each symbol $\ell \in \mL$ represents a concept with fixed semantics that all devices agree upon. This shared vocabulary provides the common reference frame necessary for multi-device semantic alignment without requiring explicit protocol negotiation.

\vspace{-2mm}\begin{definition}

For device $k$ with latent representation space $\mZ_k$, a \emph{fibrational encoder} is a functor $p^{\#}_k : \mZ_k \to \mL$.
This functor maps SRs $\bmz_k \in \mZ_k$ to semantic symbols $\ell \in \mL$, and preserves morphisms: if $\varphi : \bmz_k \to \bmz'_k$ is a morphism in $\mZ_k$ representing a semantic transformation within device $k$'s space, then $p^{\#}_k$ preserves structure by mapping this morphism to $
p^{\#}_k(\varphi) : p^{\#}_k(\bmz_k) \to p^{\#}_k(\bmz'_k) \text{ in } \mL.$
\end{definition}
The fibrational encoder thereby converts device $k$'s local representation into a semantic symbol understood by the BS and other devices. 
In practical implementation, if device $k$ wants to transmit SR $\bmz_k \in \mathbb{R}^m$, the encoder classifies $\bmz_k$ into one of the $M$ semantic symbols through 
\vspace{-2mm}\beq
p^{\#}_k(\bmz_k) = \arg \max_{\ell \in L} p(\ell \mid \bmz_k),
\vspace{-1mm}\eeq
where $p(\ell \mid \bmz_k)$ is a learned probability distribution over $\mL$.
\vspace{-2mm}\begin{definition}
    The \emph{fiber} of a symbol $\ell \in \mL$ under $p^{\#}_k$ is the set 
$
(p^{\#}_k)^{-1}(\ell) = \{ \bmz_k \in \mZ_k : p^{\#}_k(\bmz_k) = \ell \},$
which contains all SR from device $k$ that map to the same semantic symbol $\ell$. 
\vspace{-2mm}\end{definition}
The fibrational structure on $\mZ_k$ over $\mL$ is the collection of all fibers 
$
\{ (p^{\#}_k)^{-1}(\ell) : \ell \in \mL \},
$
partitioning $\mZ_k$ into equivalence classes, one per symbol in $\mL$. 
A key property termed \textit{fiber consistency} states that if $\bmz_k$ and $\bmz'_k$ both belong to $(p^{\#}_k)^{-1}(\ell)$, then any morphism $\bmz_k \to \bmz'_k$ in $\mZ_k$ represents semantically synonymous transformations within the same equivalence class $\ell$
, ensuring that the fiber structure respects the semantic organization of $\mZ_k$.
When $K$ devices transmit representations $\{ \bmz_1, \dots, \bmz_K \}$, each device $k$ encodes 
$
\bmz_k \mapsto p^{\#}_k(\bmz_k) = \ell_k \in \mL.$
The BS receives symbols $\{ \ell_1, \dots, \ell_K \} \subseteq \mL$ and must compose these symbols to infer $\bmy$. The fundamental challenge is ensuring that the composition 
$
\{ \ell_1 \otimes \ell_2 \otimes \dots \otimes \ell_K \}
$
is semantically consistent, meaning the combination respects the underlying meanings of each device's representation and the task-relevant relationships encoded in the Grothendieck topology. 
This consistency is achieved through \textit{fiber products}, which provide the categorical construction for combining aligned data in a semantically coherent way.
\vspace{-2mm}\begin{lemma}\label{lemma_fiberproduct}
For $K$ devices with fibrational encoders $p_k^{\#}: \mathcal{Z}_k \to \mathcal{L}$, $k \in \{1,\ldots,K\}$, define the fiber product as
\begin{align}
\mathcal{Z}_{\{1\ldots K\}} &= \mathcal{Z}_1 \times_{\mathcal{L}} \mathcal{Z}_2 \times_{\mathcal{L}} \cdots \times_{\mathcal{L}} \mathcal{Z}_K \nonumber \\
&= \{(z_1,\ldots,z_K) : z_k \in \mathcal{Z}_k,\ p_1^{\#}(z_1) = \cdots = p_K^{\#}(z_K)\},
\end{align}
with projections $\pi_k : \mathcal{Z}_{\{1\ldots K\}} \to \mathcal{Z}_k$. Then the following hold:
\begin{enumerate}
    \item \textbf{Semantic alignment:} Any tuple $(z_1,\ldots,z_K) \in \mathcal{Z}_{\{1\ldots K\}}$ represents a mutually aligned $K$-tuple in which all devices' representations map to the same canonical symbol $\ell \in \mathcal{L}$, so the BS interprets them as referring to a common underlying concept.
    \item \textbf{Universal property:} For any compatible tuple $\{z_k \in \mathcal{Z}_k\}_{k=1}^{K}$ satisfying $p_1^{\#}(z_1) = \cdots = p_K^{\#}(z_K)$, there exists a unique element $(z_1,\ldots,z_K) \in \mathcal{Z}_{\{1\ldots K\}}$ projecting onto each $z_k$ via $\pi_k$.
    \item \textbf{Closure under composition:} If $(z_1,\ldots,z_K), (z_1',\ldots,z_K') \in \mathcal{Z}_{\{1\ldots K\}}$ and fiber-preserving morphisms $z_k \to z_k'$ exist in each $\mathcal{Z}_k$, the composed tuple remains in the fiber product over the composed symbol in $\mathcal{L}$.
\end{enumerate}
\end{lemma}
\vspace{-2mm}\begin{IEEEproof}
See Appendix~\ref{proof_lemma_fiberproduct}.
\end{IEEEproof}
Lemma~\ref{lemma_fiberproduct} means that at the BS, we can compose two semantic languages from distinct users and create a composed semantic language, which is the RX semantic interpretation space. 
Section III formalized \textit{optimal composition} using category theoretic constructs such as lenses, Grothendieck topos and fibrations. 
Next, we ask: How do devices and the BS strategically coordinate to achieve these goals?

\vspace{-4mm}\section{Computing SR and Composition Through Multiuser Language Alignment}
\label{eq_language_alignment}
Using the category-theoretic measures of Section III, we now formulate the multi-user SC problem as a strategic game between the transmitting devices and the BS. The goal is to coordinate heterogeneous devices so that the SRs $\bm{z}_k$, once composed at the RX, maximize downstream inference quality about $\bm{y}$ while minimizing communication overhead. Composition is realized through the BS correlator $n_\phi : \mathcal{C}_1 \otimes \cdots \otimes \mathcal{C}_K \to \mathcal{C}_{\text{total}}$ followed by a put operation,
$
\bm{y} = \text{put}_\psi\big(\emptyset, n_\phi(\hat{\mathcal{C}}_1, \ldots, \hat{\mathcal{C}}_K)\big),$
where $\hat{\mathcal{C}}_i$ is reconstructed from the received symbol $\hat{\ell}_i$. This generative form makes explicit that the BS synthesizes novel information from the composed concepts rather than merely reconstructing transmitted data. The problem is strategic because no device can optimize in isolation, which induces the Stackelberg structure formalized below, with devices as leaders committing to encoders first and the BS as follower composing all transmissions.

\vspace{-5mm}\subsection{Utility Functions}\vspace{-1mm}\subsubsection{TX $k$'s Utility Function}
Each TX $k$ balances four competing objectives. It maximizes the compositional semantic information $\mathbb{S}(\bm{y}; \bm{z}_k)$ in \eqref{eq_Sy_z}  while minimizing its communication cost, which by the information bottleneck principle is approximated by the representation entropy $\mathbb{S}(\mathcal{C}_k; \ell_k) = I(\mathcal{C}_k; \ell_k) \approx H(\ell_k) = -\sum_\ell p(\ell_k)\log p(\ell_k)$. It penalizes redundancy with other devices through $D_r(\bm{z}_k, \bm{z}_{-\{k\}}) = \sum_{j\neq k} I(\bm{z}_k; \bm{z}_j)$, which couples the TX problems, and it enforces sufficient compositionality via the score $C(\bm{z}_k)$ in \eqref{eq_cs_score}. Combining these, TX $k$'s utility is
\vspace{-2mm}\begin{equation}
\begin{aligned}
    &U_k(\theta_k, w_k; \theta_{-\{k\}}, \phi, \psi) = \alpha_k \cdot \mathbb{S}(\bmy; \bmz_k) 
- \beta_k \cdot \mathbb{S}(\mathcal{C}_k; l_k) 
\\&- \gamma_k \cdot D_r(\bmz_k, \bmz_{-\{k\}}) 
- \lambda_k \cdot \max(0, C_{\text{min}} - C(\bmz_k)),
\label{eq_U_k}
\end{aligned}
\vspace{-2mm}\end{equation}
where ${\theta}_k$ are the encoder parameters, ${w}_k$ the composition weights, ${\theta}_{-\{k\}}$ the other TXs' strategies, and $(\phi, \psi)$ the BS's correlator and generator parameters. The non-negative weights $\alpha_k, \beta_k, \gamma_k, \lambda_k$ set the relative importance of each term, and the final term penalizes violations of the minimum compositionality constraint $C_{\min}$.
\subsubsection{BS's Utility Function}
We  define the BS's utility as the ratio of total semantic information to reconstruction error:
\vspace{-2mm}\begin{equation}
U_{\text{BS}}(\phi, \psi; \{\bmzh_k\}_{k=1}^K) = \sum_{k=1}^K \frac{\mathbb{S}(\bmy; \bmzh_k)} {D(\phi, \psi)}.
\label{eq_U_BS}
\vspace{-2mm}\end{equation}
This ratio captures the tradeoff in SC between maximizing compositional semantic information and minimizing inference error $D(\phi, \psi)$ for $\bmy$. The BS evaluates SRs by their effectiveness in enabling accurate inference, not just their information content. The reconstruction error $D(\phi, \psi)$ 
is defined as:
\vspace{-2mm}\begin{equation}
\begin{aligned}
D(\phi, \psi) &= \mathbb{E}\big[ \| \bmy - \text{put}_{\psi}(\varnothing, n_{\phi}(\hat{\mC}_1, \dots, \hat{\mC}_K)) \|^2 \big]
\\&= \mathbb{E}\big[ \| \bmy - \arg \min_{\hat{\bmy}} \sum_{k=1}^K E_{\psi}(\hat{\bmy}; \bmz_k) \|^2 \big]
\vspace{-2mm}\end{aligned}
\label{eq_D}
\vspace{-3mm}\end{equation}
where $E_{\psi}(\hat{\bmy}; \bmz_k)$ is an energy function quantifying the semantic compatibility between the generated inference $\hat{\bmy}$ and the received representations $\bmz_k$.
We define this energy function through a probabilistic alignment measure:
\vspace{-2mm}\begin{equation}
\begin{aligned} \small
    E_{\psi}(\hat{\bmy}; \bmz_k) = -\log \left( \frac{\exp\big(-\text{KL}(p(\bmz_k \mid \mathcal{C}_k) \,\|\, p_{\psi}(\hat{\bmy} \mid \hat{\mC}))\big)}{\sum\limits_{j=1}^K \exp\big(-\text{KL}(p(\bmz_j \mid \mC_j) \,\|\, p_{\psi}(\hat{\bmy} \mid \hat{\mC}))\big)} \right)
    \label{eq:energy_function}
\end{aligned}
\vspace{-2mm}\end{equation}
This formulation assigns low energy when device $k$'s transmitted distribution $p(\bmz_k \mid \mathcal{C}_k)$ is well-aligned with the BS's target distribution $p_{\psi}(\hat{\bmy} \mid \hat{\mC})$, and normalizes across all devices to ensure proper probability scaling. The Kullback-Leibler divergence term measures the semantic distance between the transmitted and target distributions, with smaller divergence indicating better alignment.

\vspace{-3mm}\subsection{Stackelberg Game Formulation}\vspace{-1mm}

The multi-device CSC problem exhibits a natural hierarchical 
structure that we formalize as a Stackelberg game~\cite{zhang2009stackelberg}. Devices 
(leaders) must commit to encoding strategies $(\theta_k, w_k)$ and SRs $\bmz_k$ before observing the BS's composition approach. The BS 
(follower) then 
optimally composes $K$ transmissions $\{\bmz_1, \ldots,\bmz_k\}$  via strategy $(\phi, \psi)$. Further, we formally defines 
the Stackelberg game as follows.

\vspace{-2mm}\begin{definition}  The \emph{CSC game} $\Gamma = (N, \{\Theta_k\}_{k=1}^K, \Phi, \{U_k\}_{k=1}^K, U_{\text{BS}})$ consists of:
\begin{itemize}
    \item A set of players $N = \{1, \dots, K, \text{BS}\}$
    \item Strategy spaces $\Theta_k = \{(\theta_k, w_k): \theta_k \in \mathbb{R}^{|\theta|}, w_k \in \Delta^{D-1}\}$ for each TX $k$, where $\Delta^{D-1}$ is the $(D-1)$-dimensional probability simplex
    \item Strategy space $\Phi \!=\! \{(\phi, \psi): \phi \in \mathbb{R}^{|\phi|}, \psi \in \mathbb{R}^{|\psi|}\}$ for BS
    \item Utility functions $U_k$ and $U_{\text{BS}}$ as defined in \eqref{eq_U_k} and \eqref{eq_U_BS}.
\end{itemize}
\end{definition}
The game proceeds in two stages. In stage 1, each device $k \in \mK$ simultaneously solves:
\vspace{-1mm}\begin{equation}
\begin{aligned}
\max_{\theta_k, w_k}& \; U_k(\theta_k, w_k; \theta_{-\{k\}}, \phi^*(\bmz_1, \dots, \bmz_K), \psi^*(\bmz_1, \dots, \bmz_K)), \\[-2mm]
\mbox{s.t.}\,\,
&\bmz_k = \sum_{i=1}^D w_{\{k,i\}} \, R(\text{Enc}_{\theta_k}(X_k)[i]) \\
&C(\bmz_k) \geq C_{\text{min}}, \\
&\sum_{i=1}^D w_{\{k,i\}} = 1, \quad w_{\{k,i\}} \geq 0.
\label{eq:transmitter_problem}\vspace{-2mm}
\end{aligned}
\vspace{-2mm}\end{equation}
Here, $\phi^*(\cdot)$ and $\psi^*(\cdot)$ represent the TX's anticipation of the BS's best response function, which maps the transmitted SRs to optimal composition and generation parameters.
In stage 2, after observing all SRs $\{\bmz_1, \dots, \bmz_K\}$, the BS solves:
\begin{equation}
\begin{aligned}
&\max_{\phi, \psi} \; U_{\text{BS}}(\phi, \psi; \{\bmz_k\}_{k=1}^K)\\
\text{s.t.}\,\,  &n_{\phi} \circ n_{\phi} = n_{\phi} \quad \text{(idempotence)} \\
&n_{\phi} \text{ is a magma/comagma homomorphism} \\
&\text{put}_{\psi} \text{ satisfies lens laws}\, (L_1'), (L2'), (L3') 
\label{eq:bs_problem}
\end{aligned}
\end{equation}
The constraints in \eqref{eq:bs_problem} ensure the BS's  strategy respects the categorical structures 
of Section III. The idempotence constraint forces  correlator $n_{\phi}$ to project onto valid concept combinations, the homomorphism properties preserve the compositional algebra, and the lens laws ensure $\text{put}_\psi$ generates inferences correctly from  composed concepts.
Because the redundancy term $D_r(\bmz_k, \bmz_{-\{k\}}) $ couples each TX's utility to the others' strategies, the device subproblems form a Nash game in which each device anticipates the others, while the Stackelberg structure additionally requires devices to anticipate the BS's composition strategy $(\phi^*, \psi^*)$.   
$
((\theta_k^*, w_k^*)_{k=1}^K, \phi^*, \psi^*)$
constitutes a \emph{SE} \cite{elrahi2019managing} if each device $k$ plays a best response 
   given other devices' strategies and the BS's best response: $
    (\theta_k^*, w_k^*) \in \arg \max_{\theta_k, w_k} U_k(\theta_k, w_k; \theta_{-\{k\}}^*, \phi^*(\cdot), \psi^*(\cdot))
    $, subject to constraints in \eqref{eq:transmitter_problem}, and the BS optimizes anticipating the 
   followers' equilibrium strategies:
$
    (\phi^*, \psi^*) \in \arg \max_{\phi, \psi} U_{\text{BS}}(\phi, \psi; \{\bmz_k^*\}_{k=1}^K)$, subject to constraints in \eqref{eq:bs_problem}. When the followers' Nash equilibrium  (NE) is unique, the 
SE reduces to finding the BS's best response 
to this unique equilibrium point.
This equilibrium reflects the hierarchical protocol, where TXs commit to encoding strategies anticipating how the BS will  compose their information with others.

\vspace{-6mm}\subsection{Existence and Uniqueness of Equilibrium}\vspace{-1mm}

Before presenting solution algorithms, we establish equilibrium existence through a two-step approach: 
first proving that the devices' NE exists for any 
fixed BS strategy, then showing that the overall Stackelberg 
equilibrium exists.
\begin{theorem} Given a fixed BS strategy $(\Phi,\psi)$, the devices' subgame admits at least 
one pure-strategy NE if:
\begin{enumerate}
    \item The strategy spaces $\Theta_k$ for all $k$ and $\Phi$ are non-empty, compact, and convex subsets of Euclidean space
    \item Each device's utility $U_k(\theta_k,w_k; \theta_{-k}, \Phi, \psi)$ is continuous in all 
   arguments and quasi-concave in $(\theta_k,w_k)$.
\end{enumerate}
\label{thm:existence}
\end{theorem}
\vspace{-2mm}\begin{IEEEproof}
See Appendix~\ref{proof_thm:existence}.    
\end{IEEEproof}

While Theorem \ref{thm:existence} guarantees existence under mild regularity conditions, uniqueness requires stronger assumptions on the structure of the utility functions such as strict concavity of $U_{BS}$, diagonal dominance of best-response Jacobians, and super-additive compositional information, it is neither required for our algorithmic approach nor guaranteed in practical SC systems. Our ADMM algorithm (Section IV-E) converges to stable equilibria without requiring uniqueness.
 

\vspace{-2mm}\begin{lemma}
The BS's best-response mapping $\text{BR}_{\text{BS}}: \prod_{k=1}^K \Theta_k \to \Phi$ is well-defined and upper hemicontinuous.
\end{lemma}
\vspace{-3mm}\begin{proof}
By the Weierstrass Extreme Value Theorem, the continuous function $U_{\text{BS}}(\phi,\psi;\{z_k\})$ attains its maximum on the compact set $\Phi$. Upper hemicontinuity follows from Berge's Maximum Theorem \cite{Walker1979GeneralizationMaximum}. See Appendix C for details.
\end{proof}
\vspace{-4mm}\begin{theorem}
\label{thm:stackelberg_existence}
Under the conditions of Theorem 2 and Lemma 4, the CSC Stackelberg game $\Gamma$ admits at least one equilibrium.
\end{theorem}

\begin{proof}
See Appendix~\ref{proof_thm:stackelberg_existence}. 
\end{proof}
\vspace{-7mm}\subsection{Solution via Backward Induction}
\label{subsec:backward_induction}
\vspace{-1mm}
We solve the Stackelberg game using backward induction, beginning with the BS's problem in Stage 2 and then addressing the TXs' problem in Stage 1.

\vspace{-0mm}\subsubsection{Stage 2: BS's Optimal Strategy}

The BS's problem \eqref{eq:bs_problem} is a fractional program due to the ratio form of its utility \eqref{eq_U_BS}. We employ Dinkelbach's method \cite{dinkelbach1967nonlinear}, which transforms the fractional program into a sequence of parameter-dependent subproblems.
Introducing an auxiliary variable $\beta$, we solve
\vspace{-3mm}\begin{equation}
(\phi^{(t)}, \psi^{(t)}) = \argmax_{\phi,\psi} \left[ \sum_{k=1}^K \mathbb{S}(\bm{y}; \bm{z}_k) - \beta^{(t-1)} \cdot D(\phi, \psi) \right]
\label{eq:dinkelbach_subproblem}
\vspace{-2mm}\end{equation}
followed by the update:
\vspace{-4mm}\begin{equation}
\beta^{(t)} = \frac{\sum_{k=1}^K \mathbb{S}(\bm{y}; \bm{z}_k)}{D(\phi^{(t)}, \psi^{(t)})}
\label{eq:dinkelbach_update}
\vspace{-3mm}\end{equation}
The sequence $\{\beta^{(t)}\}$ is monotonically non-decreasing and converges to the optimal $\beta^*$, at which point $(\phi^{(t)}, \psi^{(t)})$ solves the original problem. 
\vspace{-2mm}\begin{proposition}
\label{prop:bs_optimality}
The optimal composition distribution $p^*_\psi(\hat{\bm{y}} \mid \hat{\mathcal{C}})$ satisfies the first-order condition $\nabla_{p_\psi(\hat{\bm{y}} \mid \hat{\mathcal{C}})} U_{\text{BS}} = 0$, obtained by differentiating the Dinkelbach objective through the energy function \eqref{eq:energy_function} via the chain rule.
\end{proposition}

\begin{figure*}
\begin{align}\small
&\frac{\sum_{j=1}^K \exp(-\KL(p(\bm{z}_j|\mathcal{C}_j) \| p_\psi(\hat{\bm{y}}|\hat{\mathcal{C}})))}{\exp(-\KL(p(\bm{z}_k|\mathcal{C}_k) \| p_\psi(\hat{\bm{y}}|\hat{\mathcal{C}})))} \cdot \left(1 - \sum_{j \neq k} \frac{\exp(-\KL(p(\bm{z}_j|\mathcal{C}_j) \| p_\psi(\hat{\bm{y}}|\hat{\mathcal{C}})))}{\exp(-\KL(p(\bm{z}_k|\mathcal{C}_k) \| p_\psi(\hat{\bm{y}}|\hat{\mathcal{C}})))}\right)^2 \cdot \frac{p(\bm{z}_k|\mathcal{C}_k)}{p_\psi(\hat{\bm{y}}|\hat{\mathcal{C}})} \nonumber \\
&+ \beta^{(t-1)} \cdot \sum_{k=1}^K \E_{p(\hat{\bm{z}}_k|\bm{z}_k)} \left[ \sum_{\hat{c}_k} p(\hat{c}_k|\hat{\bm{z}}_k) \cdot (1 + \log p_\psi(\hat{\bm{y}}|\hat{\mathcal{C}})) \cdot \sum_{c \in \bar{\mathcal{C}}} p(p(\hat{c}_k) = \bm{y}) \right] = 0
\label{eq:bs_optimality_condition}
\end{align}
\vspace{-6mm}\end{figure*}
\vspace{-2mm}
We solve this condition numerically. Under a Gaussian assumption, $p_\psi(\hat{\bm{y}} \mid \hat{\mathcal{C}}) = \mathcal{N}(\bm{\mu}_\theta, \Sigma_\theta)$ and $p(\bm{z}_k \mid \mathcal{C}_k) = \mathcal{N}(\bm{\mu}_k, \Sigma_k)$, the KL terms admit a closed form, and solving for $\bm{\mu}_\theta$ yields the weighted average
\begin{equation}
\begin{aligned}
\bm{\mu}^*_\theta &= \sum_{k=1}^K \alpha_k \bm{\mu}_k, \\  \alpha_k &= \frac{\exp(-\text{KL}(p(\bm{z}_k \mid \mathcal{C}_k) \,\|\, p_\psi(\hat{\bm{y}} \mid \hat{\mathcal{C}})))}{\sum_{j=1}^K \exp(-\text{KL}(p(\bm{z}_j \mid \mathcal{C}_j) \,\|\, p_\psi(\hat{\bm{y}} \mid \hat{\mathcal{C}})))}.
\end{aligned}
\end{equation}
The BS thus combines received representations by their KL alignment with the target, weighting more heavily those closer to the desired remote inference output  distribution.

\subsubsection{Stage 1: TXs' Optimal Strategies}

Each device $k$ must solve \eqref{eq:transmitter_problem} while anticipating the BS's best response derived above. Since the BS's strategy $(\phi^*, \psi^*)$ depends on the transmitted representations from all devices, the TXs' problems are coupled through the reconstruction error implicitly appearing in utilities via the BS's anticipated behavior.
We characterize the optimal TX strategies by leveraging Karush-Kuhn-Tucker (KKT) conditions.

\vspace{-2mm}\begin{proposition}
\label{prop:tx_optimality}
At an equilibrium, TX $k$'s encoder parameters $\bm{\theta}^*_k$ satisfy:
\vspace{-3mm}\begin{align} \small
\nabla_{\bm{\theta}_k} U_k &= \alpha_k \cdot \nabla_{\bm{\theta}_k} \mathbb{S}(\bm{y}; \bm{z}_k) - \beta_k \cdot \nabla_{\bm{\theta}_k} H(\bm{z}_k) \nonumber \\
&\quad - \gamma_k \cdot \sum_{j \neq k} \nabla_{\bm{\theta}_k} \|\Cov(\bm{z}_k, \bm{z}_j)\|^2_F \nonumber \\
&\quad - \lambda_k \cdot \mathbbm{1}_{\{\mathcal{C}(\bm{z}_k) < C_{\min}\}} \cdot \nabla_{\bm{\theta}_k} \mathcal{C}(\bm{z}_k) = \bm{0}
\label{eq:tx_optimality_condition}
\vspace{-2mm}\end{align}
where $\mathbbm{1}_{\{\mathcal{C}(\bm{z}_k) < C_{\min}\}}$ is the indicator function that equals 1 when the compositionality constraint is violated and else 0.
\end{proposition}

The gradients derived using \eqref{eq:tx_optimality_condition} (whose expressions are omitted due to lack of space) enable us to implement gradient-based optimization algorithms for finding TX equilibrium strategies, which we detail in the next subsection.

\vspace{-5mm}\subsection{Distributed Algorithm via ADMM}
\label{subsec:admm_algorithm}\vspace{-1mm}

While the gradient conditions \eqref{eq:tx_optimality_condition} characterize the equilibrium strategies, computing them requires coordination among TXs due to the coupling through the redundancy term $D_r(\bm{z}_k, \bm{z}_{-k})$ and the BS's anticipated response. Direct implementation would require each TX to know the strategies of all others, which is impractical in distributed systems. We therefore develop a distributed algorithm based on the ADMM that enables TXs to find equilibrium strategies through limited message passing with the BS. The key idea is to introduce a consensus variable $\bar{\bm{z}}$ representing the average transmitted representation, and reformulate each TX's problem to include a penalty for deviation from this consensus. The augmented utility for TX $k$ becomes:
\vspace{-2mm}\begin{equation}
U_k^{\prime}(\bm{\theta}_k, \bm{w}_k) = U_k(\bm{\theta}_k, \bm{w}_k; \bar{\bm{z}}, \phi, \psi) + \bm{\lambda}_k^\top(\bm{z}_k - \bar{\bm{z}}) - \frac{\rho}{2}\|\bm{z}_k - \bar{\bm{z}}\|^2
\label{eq:augmented_utility}
\end{equation}
where $\bm{\lambda}_k$ is a Lagrange multiplier vector enforcing consensus, and $\rho > 0$ is a penalty parameter controlling the strength of the consensus enforcement. This augmented formulation allows each TX to optimize locally while coordinating with others through the shared variable $\bar{\bm{z}}$, which is maintained and broadcast by the BS. The ADMM algorithm alternates between three steps: (i) each TX updates its strategy to maximize its augmented utility via gradient ascent, (ii) the BS updates the consensus variable and its own strategy, and (iii) the dual variables are updated to enforce consensus constraints. We  present the complete algorithm in Algorithm~\ref{alg:admm}. 
\begin{figure*}
\begin{align}\small 
&(\bm{\theta}_k^{(t)}, \bm{w}_k^{(t)}) = \argmax_{\bm{\theta}_k, \bm{w}_k} \Big[ U_k(\bm{\theta}_k, \bm{w}_k; \bm{z}_{-k}, \phi^{(t-1)}, \psi^{(t-1)}) \quad + (\bm{\lambda}_k^{(t-1)})^\top(\bm{z}_k - \bar{\bm{z}}^{(t-1)}) - \frac{\rho}{2}\|\bm{z}_k - \bar{\bm{z}}^{(t-1)}\|^2 \Big] \\
&\text{s. t. }\, \bm{z}_k = \sum_{i=1}^D w_{k,i} R(\text{Enc}_{\bm{\theta}_k}(\bm{X}_k)[i]), \,\, 
 \mathcal{C}(\bm{z}_k) \geq C_{\min}, \,\, \sum_{i=1}^D w_{k,i} = 1, w_{k,i} \geq 0.
\label{eq_solve1}
\end{align}\vspace{-6mm}
\end{figure*}
\begin{algorithm}[t]
\caption{Distributed CSC via ADMM}
\label{alg:admm}
\begin{algorithmic}[1]\footnotesize
\STATE \textbf{Input:} $\{\bm{X}_k\}_{k=1}^K$, $D$, weights $\{\alpha_k,\beta_k,\gamma_k,\lambda_k\}$, $C_{\min}$, penalty $\rho$, tolerances $\tau_{\text{p}},\tau_{\text{d}}$, max iterations $T$
\STATE \textbf{Init:} $\forall k$: $\bm{\theta}_k$ random, $\bm{w}_k=(1/D)\mathbbm{1}_D$, $\bm{\lambda}_k=\bm{0}$, $\bm{z}_k=\sum_i w_{k,i}R(\text{Enc}_{\bm{\theta}_k}(\bm{X}_k)[i])$; $\phi,\psi$ random; $\bar{\bm{z}}^{(0)}=\frac{1}{K}\sum_k \bm{z}_k$
\FOR{$t = 1$ to $T$}
\STATE \textbf{Step 1 (TX, parallel):} for each $k$, set $\bm{z}_{-k}=K\bar{\bm{z}}^{(t-1)}-\bm{z}_k^{(t-1)}$, solve \eqref{eq_solve1} by gradient ascent on
\vspace{-1mm}
\begin{equation*}\scriptsize
\nabla_{\bm{\theta}_k} U_k' = \nabla_{\bm{\theta}_k}U_k + \big(\bm{\lambda}_k^{(t-1)} - \rho(\bm{z}_k-\bar{\bm{z}}^{(t-1)})\big)\tfrac{\partial \bm{z}_k}{\partial \bm{\theta}_k},
\end{equation*}
\vspace{-3mm}
\STATE \quad project $\bm{w}_k$ onto the simplex, recompute $\bm{z}_k^{(t)}$, transmit to BS
\STATE \textbf{Step 2 (BS):} given $\{\bm{z}_k^{(t)}\}$, run Dinkelbach: iterate $(\phi^{(t,s)},\psi^{(t,s)})=\argmax_{\phi,\psi}[\sum_k \mathbb{S}(\bm{y};\bm{z}_k^{(t)})-\beta^{(s-1)}D(\phi,\psi)]$ and update $\beta^{(s)}$ via \eqref{eq:dinkelbach_update} until $|\beta^{(s)}-\beta^{(s-1)}|<\tau_{\text{D}}$
\STATE \textbf{Step 3 (Dual):} $\bar{\bm{z}}^{(t)}=\frac{1}{K}\sum_k\bm{z}_k^{(t)}$; $\forall k$: $\bm{\lambda}_k^{(t)}=\bm{\lambda}_k^{(t-1)}+\rho(\bm{z}_k^{(t)}-\bar{\bm{z}}^{(t)})$; broadcast $\bar{\bm{z}}^{(t)},\{\bm{\lambda}_k^{(t)}\},(\phi^{(t)},\psi^{(t)})$
\STATE \textbf{Stop} if $\sqrt{\sum_k\|\bm{z}_k^{(t)}-\bar{\bm{z}}^{(t)}\|^2}<\tau_{\text{p}}$ and $\rho\sqrt{K}\|\bar{\bm{z}}^{(t)}-\bar{\bm{z}}^{(t-1)}\|<\tau_{\text{d}}$
\ENDFOR
\RETURN $\{(\bm{\theta}^*_k,\bm{w}^*_k)\}_{k=1}^K$, $(\phi^*,\psi^*)$
\end{algorithmic}
\vspace{-1mm}\end{algorithm}
We now analyze its computational complexity to assess their scalability to large-scale physical AI systems.

\subsubsection{Complexity of ADMM}

For each TX $k$,   computing $\mathbb{S}(\bm{y}; \bm{z}_k)$ using Monte Carlo estimation requires $\mathcal{O}(M_{\text{samples}} \cdot |\Pi| \cdot |\mathcal{C}|^2)$ operations, where $M_{\text{samples}}$ is the number of Monte Carlo samples, $|\Pi|$ is the number of correlators, and $|\mathcal{C}|$ is the total number of concepts. Finally, gradient computation via backpropagation through the encoder network requires $\mathcal{O}(|\bm{\theta}_k|)$ operations. 
For the BS side, each Dinkelbach iteration requires evaluating the energy function $E_\theta(\hat{\bm{y}}; \bm{z}_k)$ for all $K$ devices, which costs $\mathcal{O}(K \cdot m \cdot \dim(\bm{y}))$ operations. Computing the reconstruction error $D(\phi, \psi)$ via Monte Carlo sampling requires $\mathcal{O}(N_{\text{samples}} \cdot K \cdot m \cdot \dim(\bm{y}))$ operations. With $T_{\text{Dinkelbach}}$ Dinkelbach iterations per ADMM iteration, the total BS complexity per ADMM iteration is $\mathcal{O}(T_{\text{Dinkelbach}} \cdot N_{\text{samples}} \cdot K \cdot m \cdot \dim(\bm{y}))$.
The full ADMM algorithm requires $T_{\text{outer}}$ outer iterations to converge. The total computational cost is
$\mathcal{O}(T_{\text{outer}} \cdot  T_{\text{Dinkelbach}} \cdot N_{\text{samples}} \cdot K \cdot m \cdot \dim(\bm{y}))$. The overall complexity scales linearly in all system parameters $(K, m, dim(y))$, with the remaining factors $(T_{\text{outer}}, T_{\text{Dinkelbach}}, N_{\text{samples}})$ being small constants that converge rapidly in practice \cite{dinkelbach1967nonlinear}, making the algorithm computationally tractable even for large-scale deployments.
\vspace{-6mm}\subsection{Theoretical Properties of the Equilibrium}
\label{subsec:theoretical_properties}\vspace{-2mm}

Beyond existence and uniqueness, we now establish several desirable theoretical properties of the SE.

\begin{theorem}
\label{thm:pareto_efficiency}
If the equilibrium strategies $(\bm{\theta}^*, \bm{w}^*, \phi^*, \psi^*)$ satisfy the complementarity condition:
\vspace{-2mm}\begin{equation}
\frac{\partial \mathbb{S}(\bm{y}; \bm{z}_i)}{\partial \bm{z}_j} \geq 0 \quad \text{for all } i \neq j
\label{eq:complementarity}
\vspace{-2mm}\end{equation}
indicating that concepts from different devices are synergistic rather than substitutable, then the equilibrium is Pareto efficient.
\end{theorem}
\begin{IEEEproof}
See Appendix~\ref{proof_pareto}.
\end{IEEEproof}

\vspace{-4mm}\section{Simulation Results and Analysis}\vspace{-1mm}
We validate our CSC framework via a two-stage evaluation strategy. First, we use the Highway-Env intersection scenario from the Gymnasium environment \cite{towers2024gymnasium} to demonstrate compositional generalization under controlled conditions with ground-truth concept labels. Second, we use CARLA autonomous driving simulator \cite{dosovitskiy2017carla} with high-fidelity sensor models to validate practical performance on latency, bandwidth efficiency, and scalability under realistic  conditions.
\vspace{-5mm}\subsection{Compositional Generalization using GYM Environment}\vspace{-1mm}

We first evaluate CSC on the autonomous-vehicle intersection scenario (intersection-v1) [41], where the remote inference task is binary safety classification. Distributed sensors communicate their semantic observations to a remote decision-maker that determines whether the intersection is SAFE or UNSAFE to proceed. This task requires jointly reasoning over pedestrian behavior from the camera, vehicle response from the speed sensor, and proximity from the distance sensor. Unlike multi-modal fusion, which learns an entangled joint representation tied to the trained set of modalities and concept combinations, our compositional approach \emph{recombines independently-learned concepts via the categorical correlators of Section III}, enabling inference on concept combinations never seen jointly during training, which is what the novel-combination test set evaluates. The camera (V1) reports pedestrian position (5 values) and activity (4 values), the speed sensor (V2) reports vehicle speed (4 values) and acceleration (3 values), and the distance sensor (V3) reports distance to pedestrian (4 values) and time-to-collision (4 values), yielding a 6-concept semantic space with 25 discrete values encoded as one-hot features.

\vspace{-0mm}We collect 400 training scenarios through random exploration and construct a balanced  $200$ scenario test set  with $50\%$ novel combinations (concept patterns never seen in  training) and $50\%$ seen combinations. Safety follows compositional rules requiring 3-way interactions across all three sensors. For example, a running pedestrian (V1) at edge positions (V1) is safe if the vehicle is braking (V2) AND distance is medium (V3), but unsafe if accelerating (V2) AND distance is close (V3). This simple setup enables rigorous evaluation of compositional generalization by testing whether models can compose known concepts in novel ways to make correct inferences, while the balanced test set allows comparison of performance on both interpolation (seen combinations) and compositional reasoning (novel combinations).
\begin{figure*}[t]
\begin{subfigure}{.33\textwidth}\vspace{-0mm}
{\includegraphics[width=1\columnwidth]{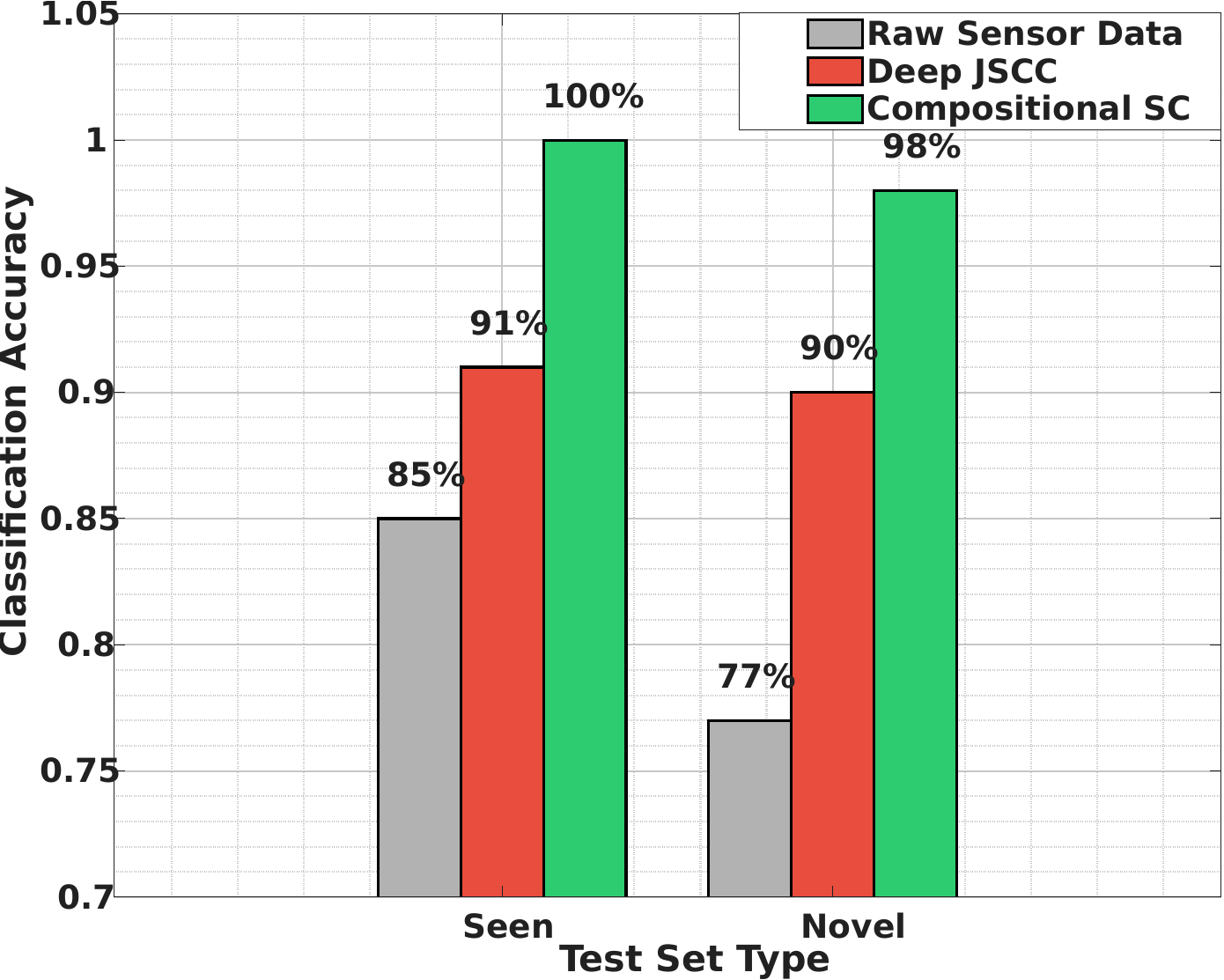}}\vspace{-1mm}
\caption{}
\label{Seen_vs_Novel}\vspace{-1mm}
\end{subfigure}
\vspace{-0mm}\begin{subfigure}{.33\textwidth}
\vspace{-0mm}{\includegraphics[width=0.9\columnwidth]{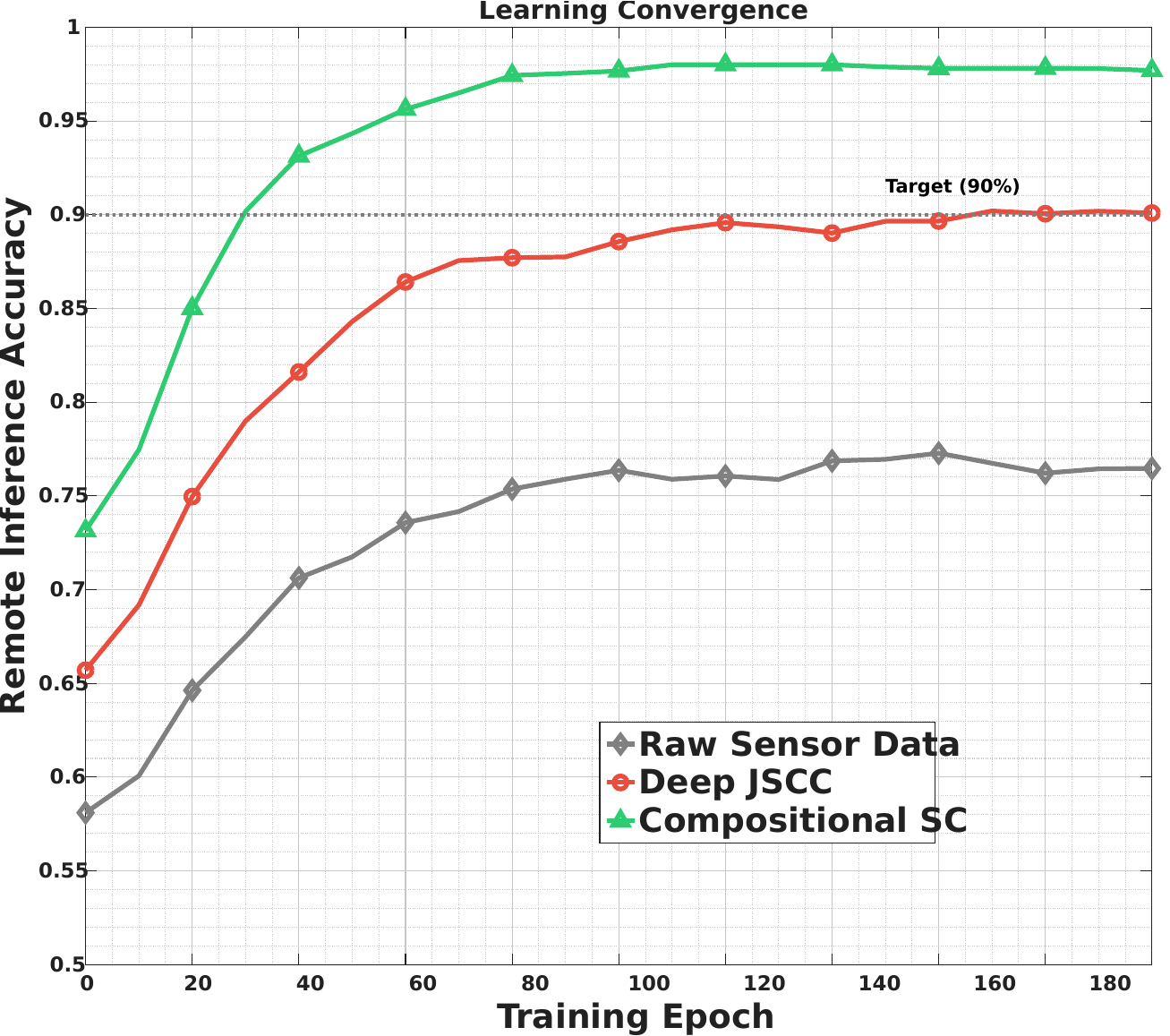}}\vspace{-1mm}
\caption{}
\label{RI_accuracy}\vspace{-1mm}
\end{subfigure}
\vspace{-3mm}\begin{subfigure}{.34\textwidth}\vspace{-2mm}
\centerline{\hspace{2mm}\includegraphics[width=0.93\columnwidth]{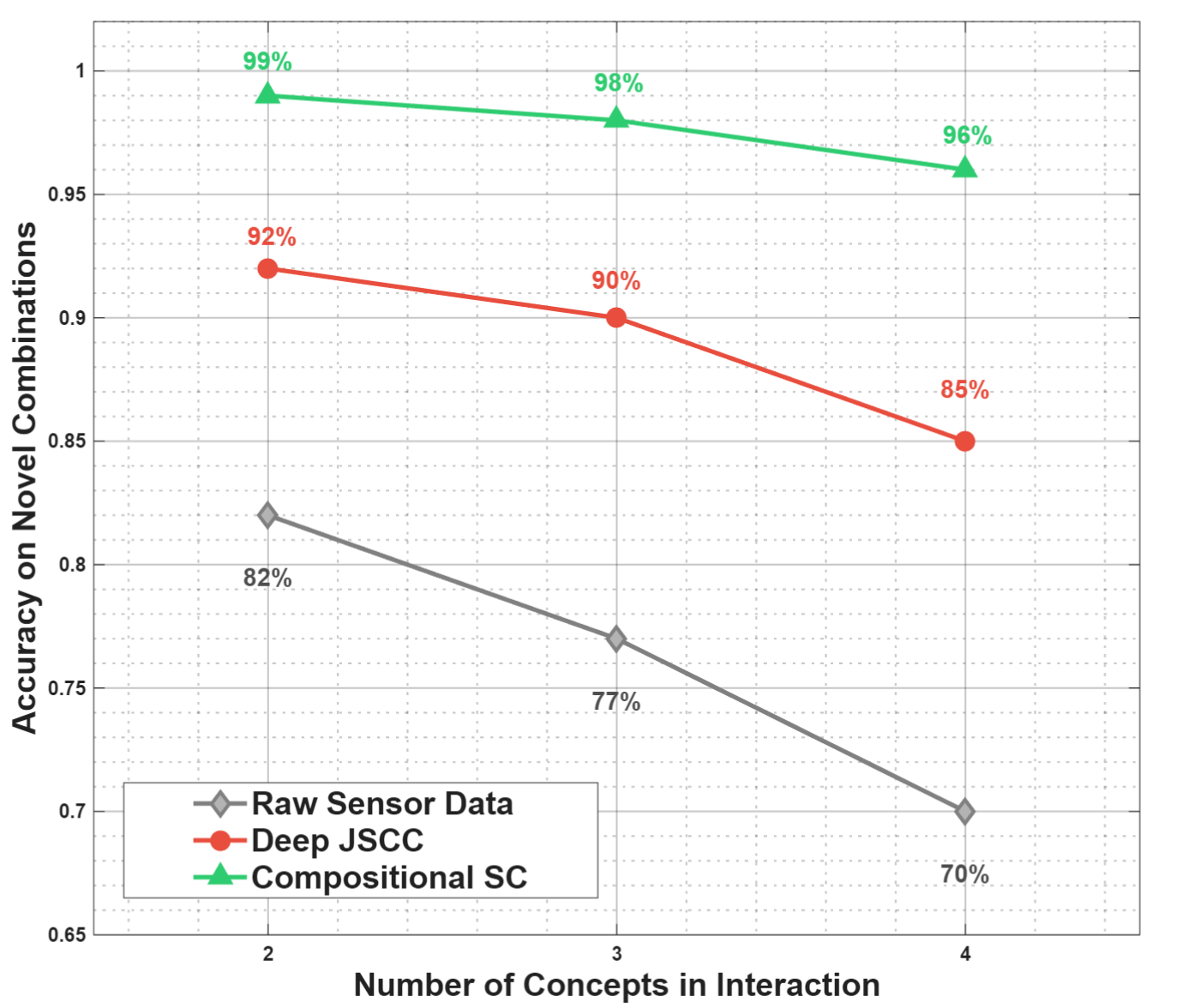}}\vspace{-2mm}
\caption{}
\label{Complexity_Robustness}\vspace{-1mm}
\vspace{-1mm}\end{subfigure}
\vspace{-4mm}\caption{\scriptsize (a) Classification accuracy for training and test distribution. (b) Proposed game-theoretic CSC Convergence vs training epoch. (c) Accuracy on novel combinations as a function of number of concepts fused, showing scalability}
\vspace{-0mm}
\label{EMLangCOnvergence}
\vspace{-7mm}\end{figure*}

\vspace{-0mm}Fig.~\ref{Seen_vs_Novel} shows that the  proposed CSC achieves the best performance in terms  of training and test classification accuracy. For novel combinations encountered during test time, the proposed method achieves gains of around $8\%$  compared to deep JSCC methods. Remote inference accuracy evolution during training phase is provided in Fig.~\ref{RI_accuracy}. Fig.~\ref{Complexity_Robustness} shows the consistent performance of proposed CSC algorithm as the number of concepts increases, while deep JSCC performance deteriorates as the number of concepts increases. This is because deep JSCC would require retraining as the number of concepts increases contrary to the proposed CSC method. Fig.~\ref{fig:numbits_sensors} shows that the number of bits communicated increases at a slower rate with the number of sensors thanks to the compositional generalization, compared to baseline methods whose compression schemes change as sensor dynamics change. Moreover, the task accuracy remains consistently high, even as the compression ratio is increased, for the proposed CSC.

\vspace{-4mm}\subsection{Using CARLA Environment}\vspace{-1mm}
We validate the game-theoretic SE framework using the CARLA simulator (version 0.9.14) \cite{dosovitskiy2017carla} with high-fidelity sensor models. The experimental setup consists of autonomous vehicles equipped with multiple sensors (front-facing RGB camera $800\times 600$ resolution, LIDAR generating $15,000$ points/scan, radar, and IMU) transmitting raw sensor data at $7.88$ MB/s to a central cloud BS. The BS performs real-time object detection using YOLO v3 architecture trained on the COCO dataset \cite{lin2014coco}, requiring inference accuracy $ \geq 85\%$ mean average precision (mAP) within $100$~ms for safety-critical decision-making. We implement four communication strategies: (1) Uniform: random concept selection (baseline), (2) Distributed gradient descent (GD) which performs decentralized optimization without coordination, (3) Cooperative multi-agent (MA), centralized with preferences aggregation, and (4) Game-theoretic CSC: our proposed Stackelberg game formulation where the BS is the leader optimizing channel allocation, and each vehicle is a follower selecting semantic concepts respecting transmission preferences. We use $10$ independent simulation runs, each with $50$ optimization iterations over $20$ communication rounds, varying network parameters (number of vehicles $1-6$, semantic concept count $4-20$, device utilization $20-100\%$) to comprehensively evaluate robustness. Performance is measured via five metrics: (a) BS utility at equilibrium, (b) compositionality score in \eqref{eq_cs_score}, (c) inference accuracy, defined as mean average precision (mAP) on real-time object detection, (d) bandwidth efficiency, defined as percentage reduction from raw $7.88$ MB/s baseline, and (e) end-to-end latency, which is the time from sensor capture to inference output. Statistical significance is established via mean $\pm$ standard deviation computed from $10$ independent runs with independent random seeds for each scenario.
\begin{figure}[t]
\centering
\begin{subfigure}[b]{0.49\columnwidth}
    \centering    \includegraphics[width=\columnwidth,height=1\columnwidth]{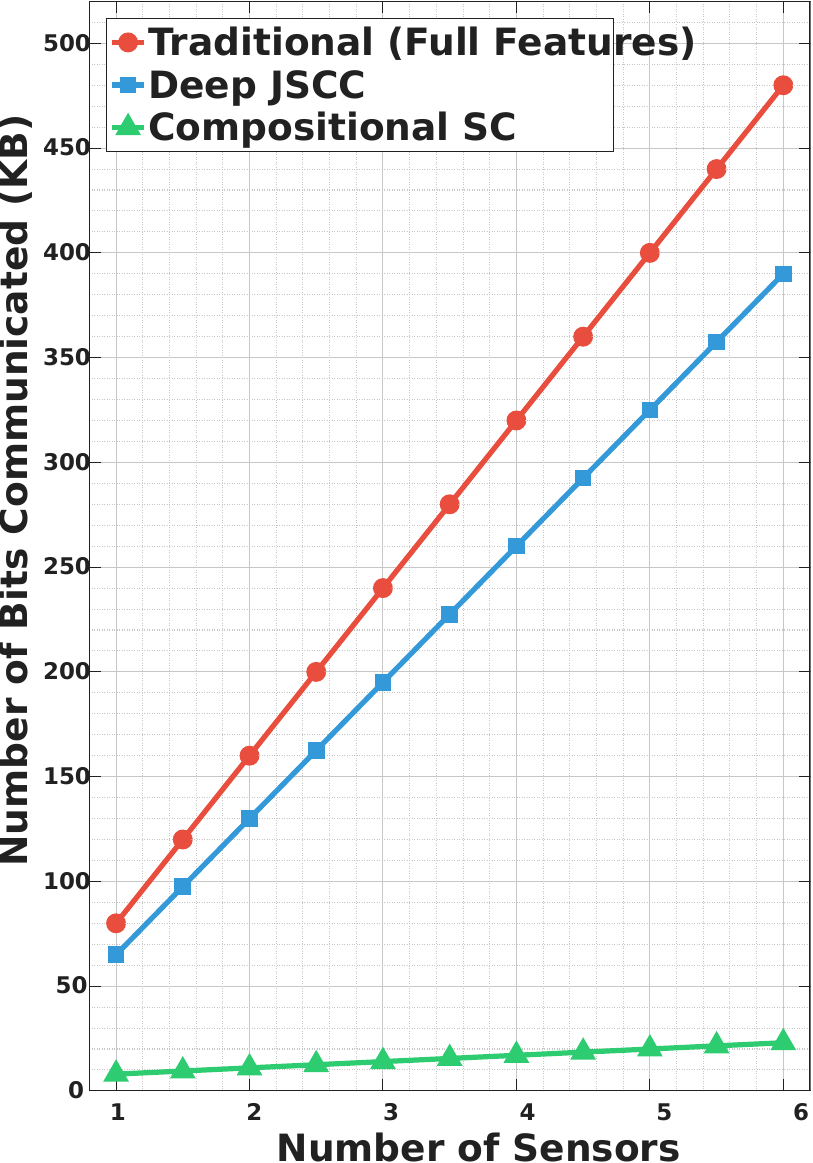}
   \vspace{-6mm} 
   \caption{\small }
\label{fig:numbits_sensors}
\end{subfigure}
\hfill
\begin{subfigure}[b]{0.49\columnwidth}
    \centering
    \includegraphics[width=\columnwidth,height=1\columnwidth]{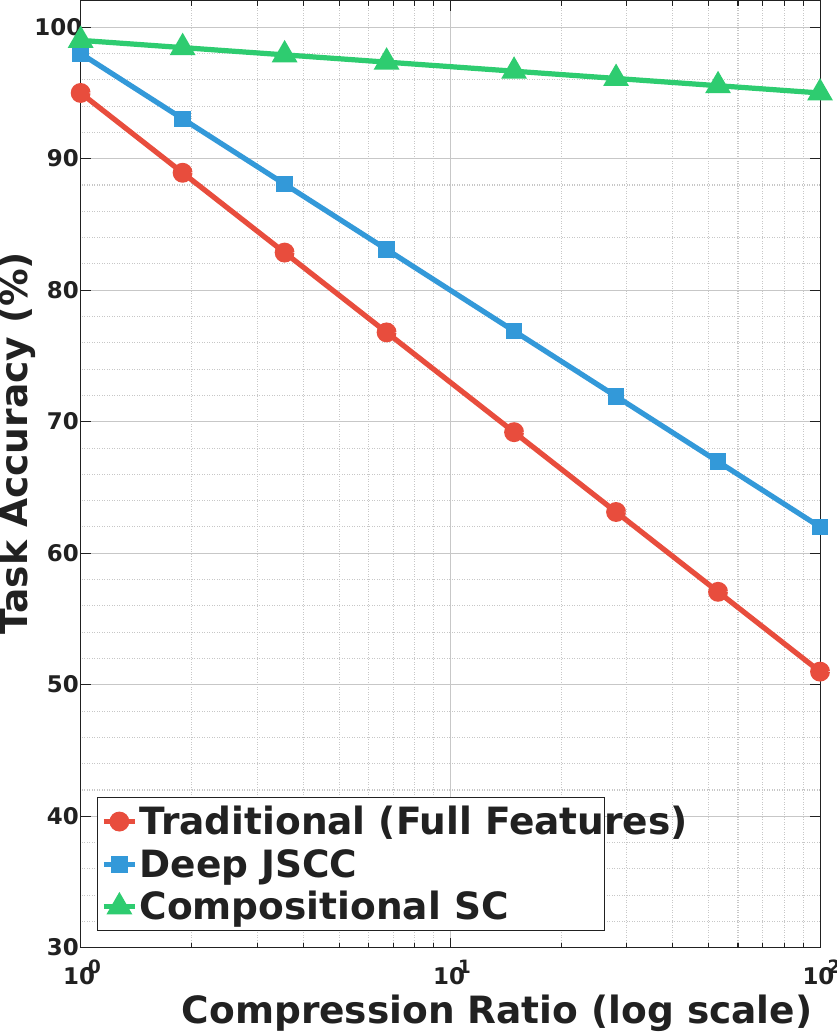}
\vspace{-6mm}\caption{\small }
    \label{fig:accuracy_compressions}
\end{subfigure}
\vspace{-7mm}\caption{\scriptsize (a) Number of bits communicated vs number of sensors, showing the scalability of proposed CSC compared to baselines. (b) Task accuracy vs compression ratio, which is the ratio of number of concepts communicated vs the total extracted semantic concepts. }
\label{fig:rewards_vs_Tmax}
\vspace{-3mm}\end{figure}

\subsubsection{Latency  improvements}
Fig.~\ref{fig:latency} shows that game-theoretic CSC achieves $70.8 \pm 3.2$ ms end-to-end latency, meeting the strict $<100$ms requirement for safety-critical autonomous driving with a $29.2$ ms safety margin providing headroom for system uncertainties and worst-case conditions. This represents a $55\%$ and $40\%$ improvement over the raw data baseline and distributed GD, respectively , demonstrating that CSC not only reduces bandwidth but actually accelerates remote inference processing through reduced data volume. This means that the proposed game-theoretic CSC enables real-time autonomous vehicle coordination while ensuring accurate  control decisions from remote inference. Moreover, the game-theoretic CSC exhibits scalability as the number of devices increases as shown in Fig.~\ref{fig:latency_devices}. Latency increases from $61.8$ ms (1 device) to $79.4$ ms ($6$ devices),  demonstrating sublinear latency growth. Cooperative MA scales from $73.2$ ms to $91.4$ ms ($+18.2$ ms), while Distributed GD ranges from $114.2$ ms to $127.8$ ms, showing that game-theoretic coordination minimizes communication overhead in multi-device systems through communicating relevant compositional semantics. 

\subsubsection{Bandwidth improvements}Fig.~\ref{fig:compositionality_concepts} shows that the game-theoretic CSC maintains $99\%$ bandwidth reduction as the number of semantic concepts increases, demonstrating strong robustness to representation richness. In contrast to the proposed CSC, baseline methods such as  distributed GD and uniform random methods show decreasing bandwidth efficiency with increased complexity. 
\begin{figure}[t]
\centering
\begin{subfigure}[b]{0.49\columnwidth}
    \centering    \includegraphics[width=0.8\columnwidth]{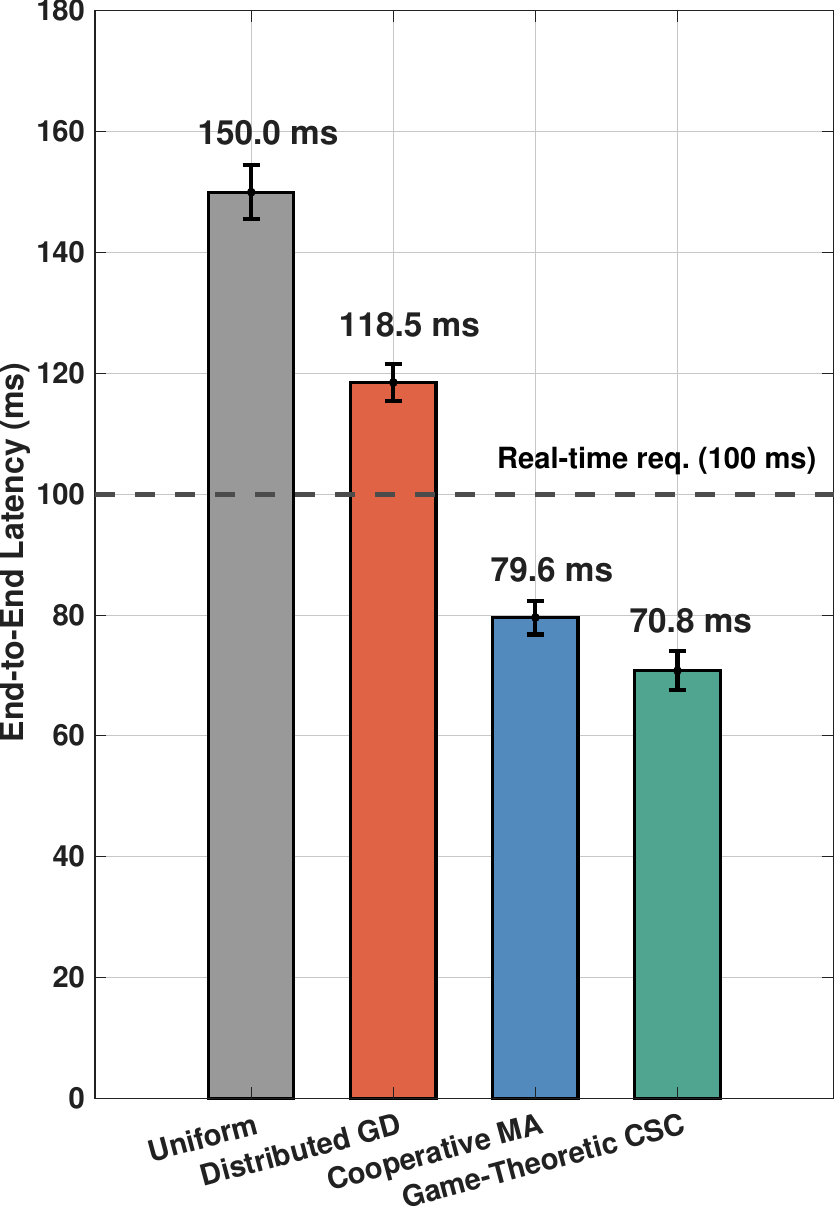}
   \vspace{-2mm} 
   \caption{\small }
\label{fig:latency}
\end{subfigure}
\hfill
\begin{subfigure}[b]{0.49\columnwidth}
    \centering
    \includegraphics[width=0.8\columnwidth]{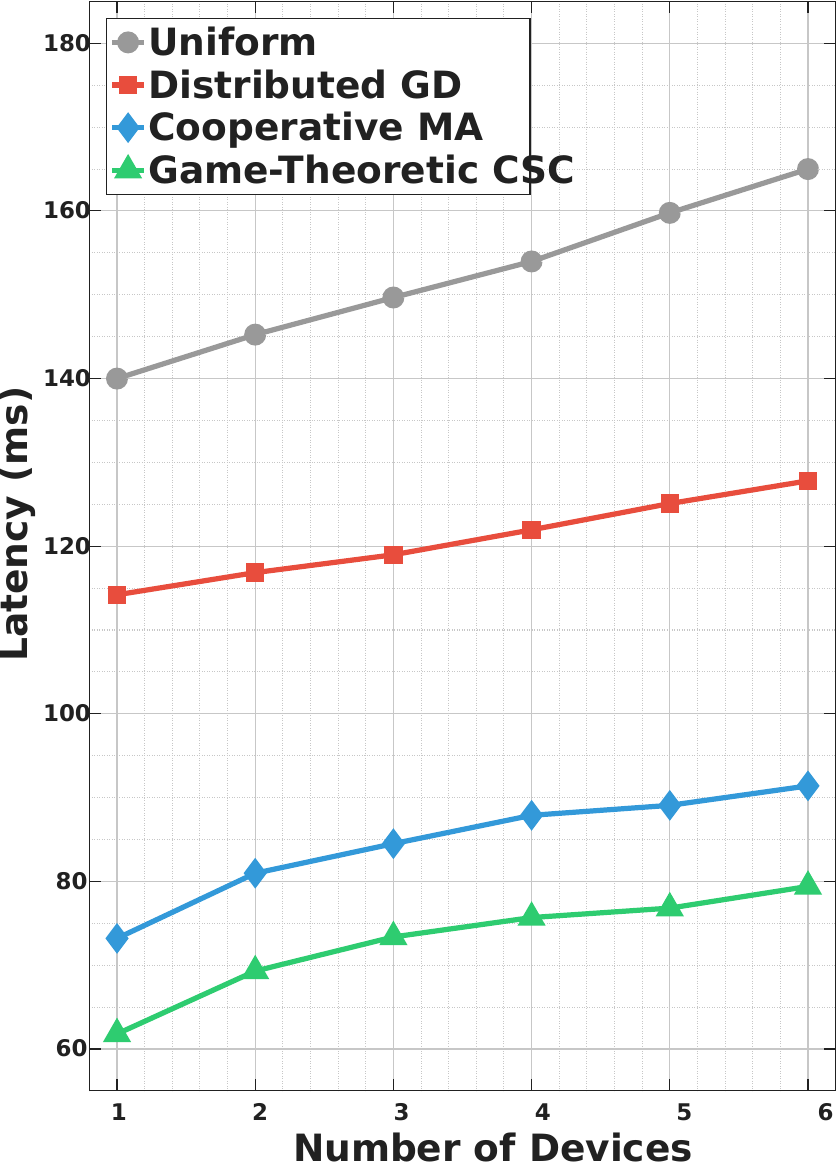}
\vspace{-2mm}\caption{\small }
    \label{fig:latency_devices}
\end{subfigure}
\vspace{-7mm}\caption{\scriptsize  (a) Average end-to-end latency for proposed CSC and baseline methods. (b) Scalability with respect to latency as the number of devices increase. }
\label{fig:rewards_vs_Tmax}
\vspace{-3mm}\end{figure}
\begin{figure}[t]
\centering
\vspace{-0mm}
\begin{subfigure}[b]{0.49\columnwidth}
    \centering
\includegraphics[width=0.94\columnwidth,height=1.2\columnwidth]{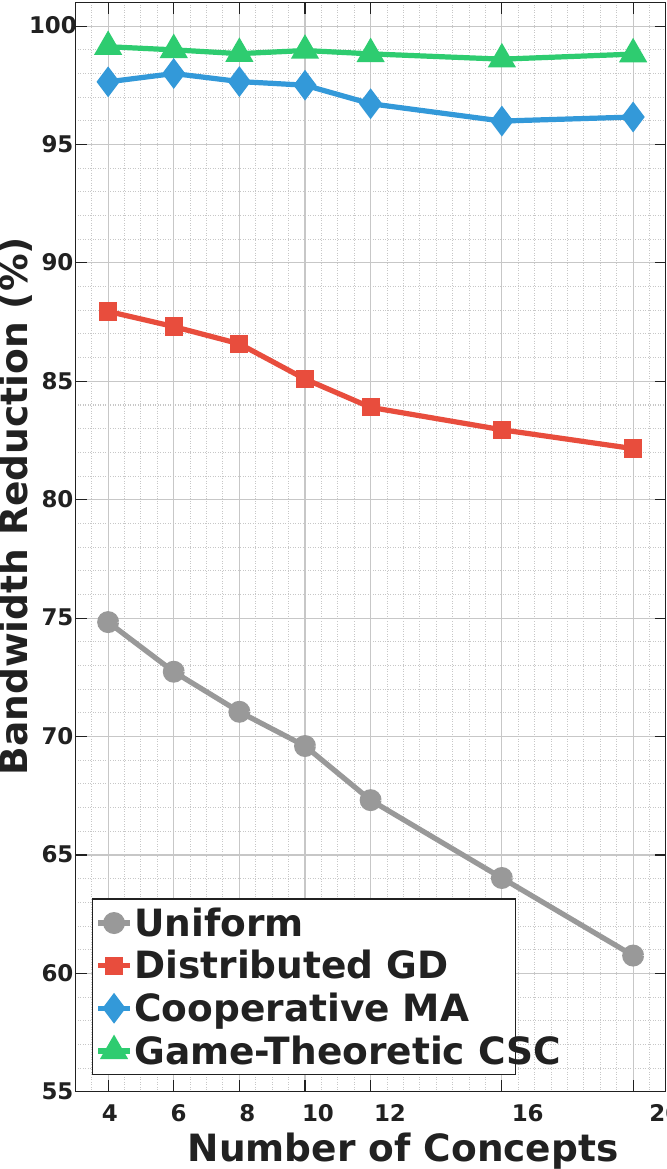}
\vspace{-2mm}\caption{\small }
    \label{fig:compositionality_concepts}
 \end{subfigure}
   \hfill
    \begin{subfigure}[b]{0.49\columnwidth}
    \centering    \includegraphics[width=0.94\columnwidth,height=1.2\columnwidth]{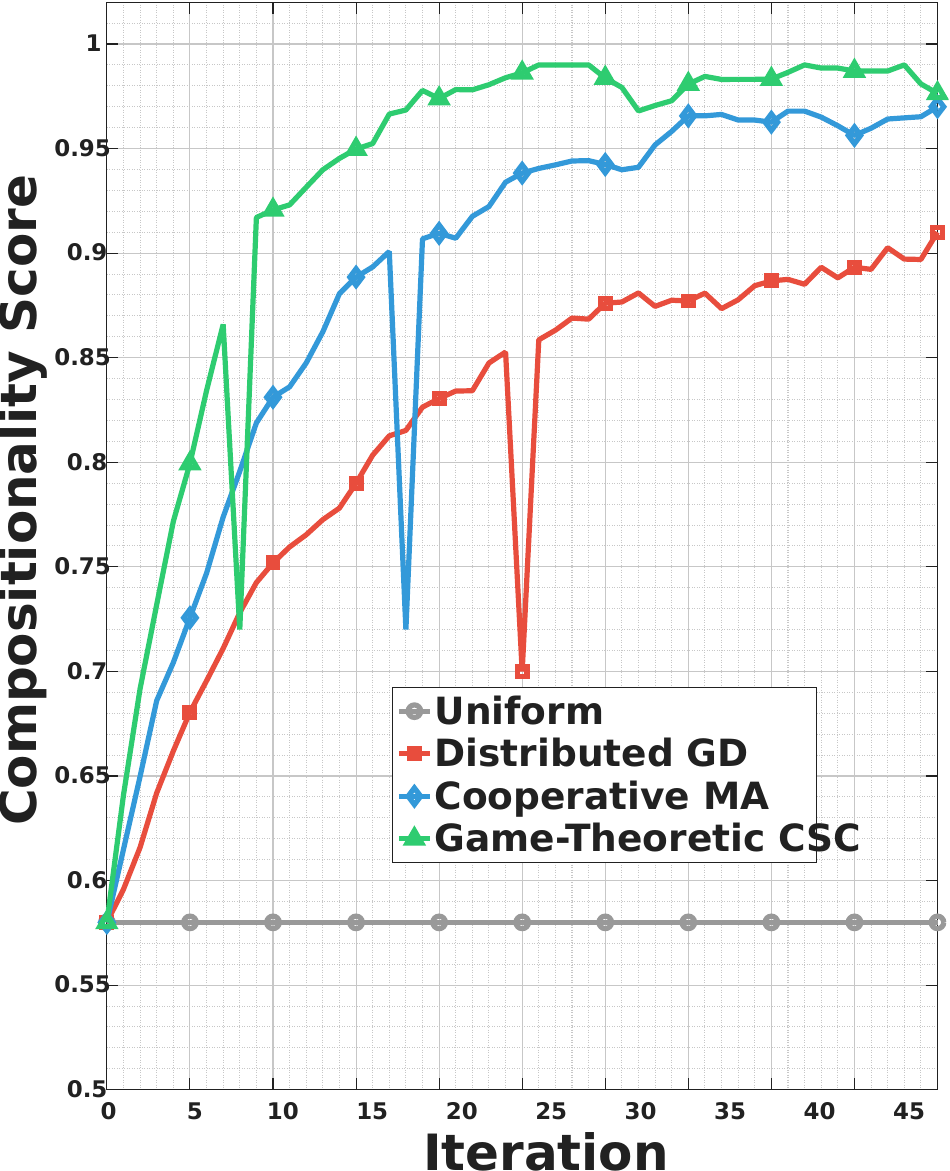}
   \vspace{-2mm} 
   \caption{\small }
\label{fig:compositionality}
\end{subfigure}
\vspace{-8mm}\caption{\scriptsize   (a) Bandwidth savings as the number of concepts. (b) Convergence of compositionality score as training progresses.}
\label{fig:rewards_vs_Tmax}
\vspace{-5mm}\end{figure}

\begin{figure}[t]
\centerline{\includegraphics[width=1\columnwidth,height=0.8\columnwidth,keepaspectratio]{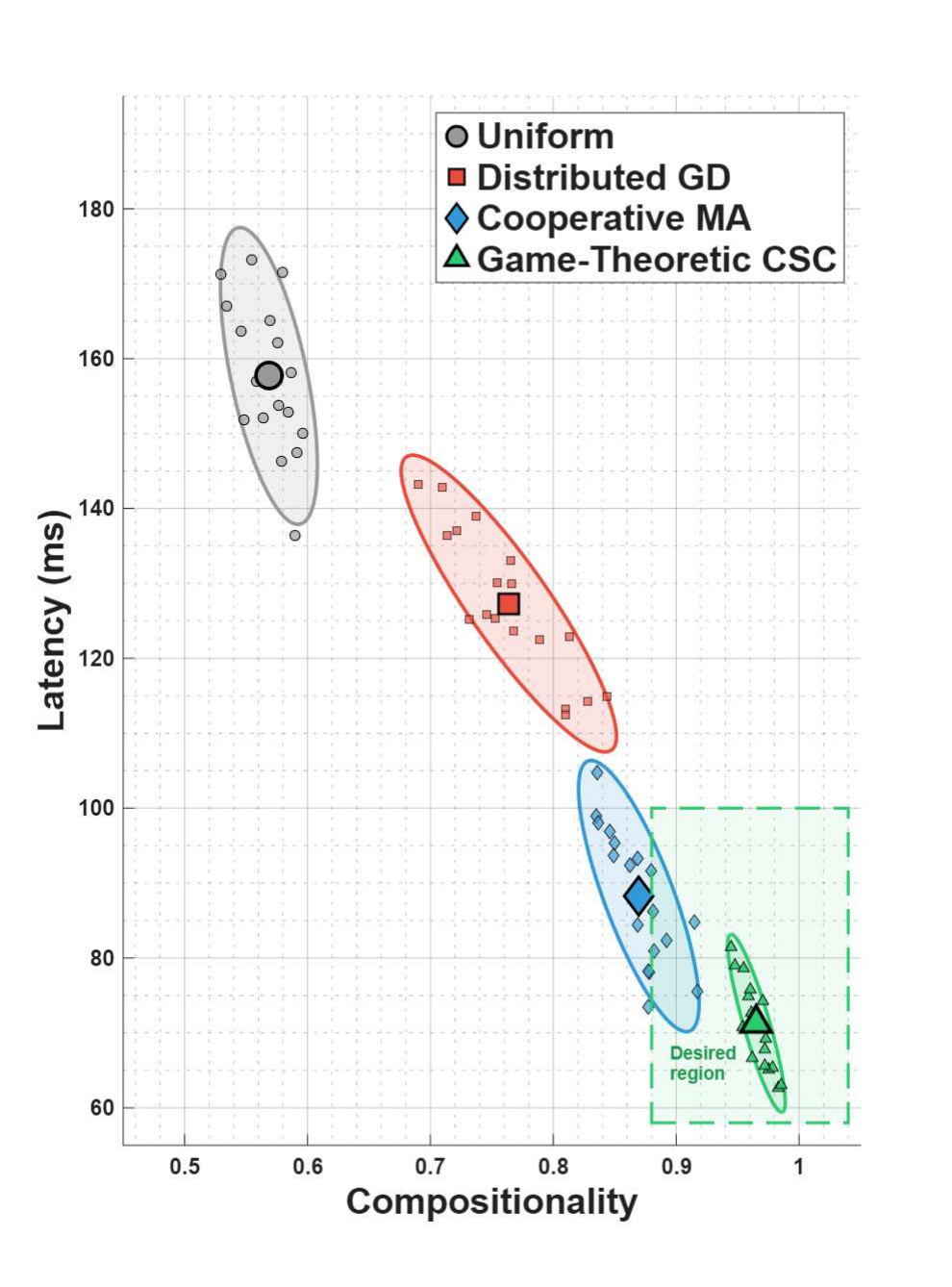}}\vspace{-5.5mm}
\caption{\small Latency vs compositionality tradeoff.}
\label{fig:latency_compositionality}\vspace{-0mm}
\vspace{-3mm}
\vspace{-0mm}\end{figure}

\subsubsection{Compositionality}
Fig.~\ref{fig:compositionality} shows that game-theoretic CSC and cooperative MA achieve a faster convergence of the compositionality score compared to other baselines. Moreover, the compositionality score of the Uniform baseline remains flat at a value of $0.58$ across all iterations and hence exhibit negligible compositionality. The faster compositionality score convergence and the resulting higher compositionality score stem from  the game-theoretic incentive structure that explicitly rewards agents for transmitting compositionally valuable semantic concepts. Unlike distributed GD which optimizes local objectives independently, or cooperative MA approaches that assume perfect cooperation, the game-theoretic formulation creates NE dynamics where each agent's best response naturally involves selecting and transmitting concepts that maximize both individual utility and system-wide compositionality. From Fig.~\ref{fig:latency_compositionality}, we can observe that game-theoretic CSC achieves superior latency-compositionality tradeoff , maintaining $68-78$ ms latency while achieving $0.8-1.0$ compositionality. This means that semantic structure preservation does not compromise real-time performance.  Cooperative MA shows moderate performance ($80-110$ ms at $0.8-0.98$ compositionality), while Distributed GD shows poor trade-off ($100-160$ ms at $0.6-0.95$ compositionality), proving that game-theoretic incentives naturally align low-latency and high-compositionality objectives. 

\vspace{-4mm}\section{Conclusion}\vspace{-2mm}

In this paper, we have introduced the first comprehensive framework for CSC in physical AI systems addressing fundamental limitations of existing SC approaches: limited real-time adaptability, poor generalization to out-of-distribution data, and scalability challenges in coordinating heterogeneous sensing sources. By transmitting only compositional SRs rather than raw sensor data, our approach achieves significantly reduced communication overhead compared to traditional methods that transmit full-dimensional sensor observations. Unlike deep JSCC and other reasoning driven SC methods, our framework provides explicit compositional structure through lens operations and Grothendieck topos. Semantic alignment across heterogeneous sensing devices is achieved via Stackelberg game formulation, which is solved through ADMM method. Extensive simulations on realistic autonomous driving scenarios have demonstrated substantial bandwidth reduction while maintaining high inference accuracy. The compositional structure is shown to enable deductive reasoning at the RX, allowing the system to handle novel combinations of concepts without retraining, addressing   generalizability and scalability challenges in data-driven DL based SC. 

\vspace{-3mm}\bibliographystyle{IEEEbib}
\def\baselinestretch{0.82}
\bibliography{refs,semantics_ref}

\vspace{-2mm}\appendices
\vspace{-5mm}\section{Proof of Theorem 1}
\label{Proof_Theorem1}
In our multi-device SC system, an object is a semantic 
concept or inference target $y$. The Maximality Axiom ensures completeness by 
guaranteeing that the maximal sieve $M_y = \{c_k \mid k \in \{1, \ldots, K\}\}$ 
contains all device representations, so every device's contribution can be 
considered without excluding any relevant information. 
The Stability Axiom ensures if $\bmz_k$ alone or with subset $\{z_j \mid j \in S\}$ is consistent for $y$, then $\bmz_k$ with any smaller subset is also consistent. This property preserves semantic 
alignment when devices dynamically join or leave the system. 
The Transitivity Axiom ensures compositional closure by multi-hop paths $z_k \to z_{\text{intermediate}} \to y$ preserve consistency. This property enables multi-hop semantic reasoning across 
devices while preserving validity. 
Together, these three axioms guarantee that composed representations from $K$ 
heterogeneous devices maintain well-defined semantics at the BS without requiring 
explicit semantic alignment protocols.
\vspace{-2mm}\section{Proof of Lemma~\ref{lemma_fiberproduct}}\vspace{-2mm}
\label{proof_lemma_fiberproduct}
For semantic alignment, by definition of the fiber product, if $(\bmz_1, \bmz_2) \in \mZ_{\{12\}}$, then $p^{\#}_1(\bmz_1) = p^{\#}_2(\bmz_2)$. Call this common symbol $\ell$. Both representations belong to the fibers $(p^{\#}_1)^{\wedge \{-1\}}(\ell)$ and $(p^{\#}_2)^{\wedge \{-1\}}(\ell)$ respectively, \textit{ensuring they are semantically aligned to the same canonical meaning}. For the universal property, suppose $\bmz_1 \in \mZ_1$ and $\bmz_2 \in \mZ_2$ satisfy $p^{\#}_1(\bmz_1) = p^{\#}_2(\bmz_2) = \ell$. By definition of the fiber product, $(\bmz_1, \bmz_2) \in \mZ_{\{12\}}$ because this condition is exactly the fiber product's defining constraint. \textit{Uniqueness follows because the fiber product is constructed precisely as the set of pairs satisfying this condition, so no other pair in $\mZ_{\{12\}}$ projects to $\bmz_1$ and $\bmz_2$ simultaneously}. For closure under composition, suppose $(\bmz_1, \bmz_2) \in \mZ_{\{12\}}$, so $p^{\#}_1(\bmz_1) = p^{\#}_2(\bmz_2) = \ell$. Let $(\bmz'_1, \bmz'_2) \in \mZ_{\{12\}}$, so $p^{\#}_1(\bmz'_1) = p^{\#}_2(\bmz'_2) = \ell'$ (possibly different from $\ell$). For composed objects $\bmz_1 \oplus \bmz'_1 \in \mZ_1$ and $\mZ_2 \oplus \bmz'_2 \in \mZ_2$, the fibrational structure is preserved when $p^{\#}_1(\bmz_1 \oplus \bmz'_1) = p^{\#}_1(\bmz_1) \otimes p^{\#}_1(\bmz'_1) = \ell \otimes \ell'$ and $p^{\#}_2(\bmz_2 \oplus \bmz'_2) = p^{\#}_2(\bmz_2) \otimes p^{\#}_2(\bmz'_2) = \ell \otimes \ell'$, where $\otimes$ denotes composition in $L$. Since both map to $\ell \otimes \ell'$, the composed pair remains in the fiber product over $\ell \otimes \ell'$, preserving the compositional structure. For $K$ devices, proof follows by induction.

\vspace{-4mm}\section{Proof of Theorem~\ref{thm:existence}}\vspace{-2mm}
\label{proof_thm:existence}
Follows from Rosen's theorem on concave games \cite{rosen1965concave}. 
The strategy spaces are compact and convex by construction (probability 
simplex $\Delta^{D-1}$). The utility functions are continuous from the continuity 
of $\mathbb{S}(\bmy;\bmz_k), \mathbb{H}(\bmz_k),$ and $\mathbb{I}(\bmz_k;\bmz_j)$. Quasi-concavity follows from the quadratic 
form of the redundancy term and the concavity of entropy. Therefore, by 
\cite[Theorem 1]{rosen1965concave}, a pure-strategy NE exists. 


\vspace{-4mm}\section{Proof of Theorem~\ref{thm:stackelberg_existence}}\vspace{-2mm}
\label{proof_thm:stackelberg_existence}
By backward induction. For any device strategies $\{\theta_k\}$, the BS's problem \eqref{eq:bs_problem} admits a solution by compactness and quasi-concavity (Lemma 4). For any BS strategy $(\phi,\psi)$, Theorem 2 guarantees a devices' NE exists. The induced game among devices, where each device anticipates $\text{BR}_{\text{BS}}$, admits a fixed point by the Debreu-Glicksberg-Fan theorem. This fixed point, combined with $\text{BR}_{\text{BS}}$, constitutes a Stackelberg equilibrium.
\vspace{-5mm}

\section{Proof of Theorem~\ref{thm:pareto_efficiency}}\vspace{-2mm}
\label{proof_pareto}

Suppose by contradiction that the equilibrium is not Pareto efficient. Then there exists an alternative strategy profile $(\tilde{\bm{\theta}}, \tilde{\bm{w}}, \tilde{\phi}, \tilde{\psi})$ such that $U_k(\tilde{\bm{\theta}}_k, \tilde{\bm{w}}_k; \cdot) \geq U_k(\bm{\theta}^*_k, \bm{w}^*_k; \cdot)$ for all $k$ with strict inequality for at least one $k$, and $U_{\text{BS}}(\tilde{\phi}, \tilde{\psi}; \cdot) \geq U_{\text{BS}}(\phi^*, \psi^*; \cdot)$. Consider the TX $k'$ for whom the inequality is strict. By the definition of equilibrium, $\bm{\theta}^*_{k'}$ maximizes $U_{k'}$ over all feasible strategies. Since $(\tilde{\bm{\theta}}_{k'}, \tilde{\bm{w}}_{k'})$ yields strictly higher utility, we must have $\bm{\theta}^*_{-k'} \neq \tilde{\bm{\theta}}_{-k'}$.
By the complementarity condition \eqref{eq:complementarity}, increasing $\mathbb{S}(\bm{y}; \bm{z}_j)$ also increases $\mathbb{S}(\bm{y}; \bm{z}_{k'})$, which would increase $U_{k'}$. However, since $(\bm{\theta}^*_{k'}, \bm{w}^*_{k'})$ was already optimal, this leads to a contradiction. Therefore, no Pareto-improving deviation exists.
\end{document}